\documentclass[12pt]{amsart}
\usepackage[all]{xy}
\usepackage{url} % hyperref works too
\usepackage{graphicx,enumerate}
\usepackage{subfigure}
\usepackage[latin1]{inputenc}
\usepackage{tkz-graph}
\theoremstyle{plain} %--default
\newtheorem{theorem}             {Theorem}  [section]

\newtheorem{proposition}[theorem]{Proposition}

\theoremstyle{definition}

\theoremstyle{remark}

\numberwithin{equation}{section}

\def\aand{\quad \textrm{and} \quad}

\begin{document}
% Fakesection

\title[]{Countering Violent Extremism: A mathematical model}

\author[]{Manuele Santoprete\\
                %Fei Xu\\
}
                \address{Department of Mathematics\\
                         Wilfrid Laurier University\\
                         Waterloo, ON, Canada}%
              %  \address{Mathematical Research Section\\
              %           School of Mathematical Sciences\\
              %           Australian National University\\
               %          Canberra ACT 2601, Australia}

%\subjclass{Primary 54C40, 14E20; Secondary 46E25, 20C20}
%\date     {July 2, 1991}
%\keywords{Pippo}
%         \thanks will become a 1st page footnote.
%\thanks{The first author was supported in part by NSF
%                Grant \#000000.}
%\dedicatory{}

\begin{abstract}
    The term radicalization refers to the process of developing extremist  religious political or social
beliefs and ideologies. Radicalization becomes a threat to national security when it leads to violence.
Prevention and de-radicalization initiatives are part of a set of strategies used to combat  violent extremism, 
which taken together are known as  Countering Violent Extremism (CVE).
Prevention programs aim to stop the radicalization process before it starts.
   De-radicalization programs  attempt to reform convicted extremists with the ultimate goal of social reintegration. 
   We describe prevention and de-radicalization programs mathematically using a compartmental model.  
   The prevention initiatives are modeled by including a vaccination compartment,
    while the de-radicalization process is modeled by including a treatment compartment.
    The model  exhibits a threshold dynamics characterized by the basic reproduction number $ R _0 $. 
    When $ R _0<  1 $ the system  has a unique equilibrium that is asymptotically stable. When $ R _0 >1 $ the system  has another equilibrium called  
    ``endemic equilibrium", which is globally asymptotically stable.
    These results are established by using  Lyapunov functions and LaSalle's invariance principle.  
    We perform numerical simulations to confirm our theoretical  results.

    %A Lyapunov function is used to show that, for $ R _0 >1 $  the endemic equilibrium is  global asymptotic
    %stable.  
  \end{abstract}
\maketitle

\begin{center} \today \end{center}

%\tableofcontents
%}
\section{Introduction}

The term {\it radicalization} refers to the process of developing extremist  religious political or social
beliefs and ideologies.  While radical thinking is by no means problematic in itself, it becomes a threat to national security when it leads to violence.  Because of this fact, radicalization is of particular concern for governments,  law enforcement and security agencies.  

A conventional, but arguably  antiquated, approach to national security is based on counterterrorism.
Counterterrorism strategies  consist of a law enforcement component (terrorists  are arrested, 
tried, and convicted) and a military component (terrorists  lose their life or are  captured on the battleground).

    Practitioners of counterterrorism, however,  agree that these approaches alone cannot break the cycle of violence \cite{selim2016approaches}.
    In light  of this, governments use an additional set of initiatives   collectively known as 
    countering violent extremism (CVE). 
    CVE programs can be classified into three categories\cite{selim2016approaches,mastroe2016surveying,clutterbuck2015deradicalization}
    \begin{enumerate} 
        \item {\it Prevention programs}, which aim to stop the radicalization process before it starts;
            %seek to prevent the radicalization process from occurring and taking hold in the first place;
        \item {\it Disengagement programs}, which  endeavor to block radicalization while it is taking place. 
            %to stop or control radicalization as it is occurring;
        \item {\it De-radicalization programs}, which  aim to reform convicted extremists with the ultimate goal of 
             social reintegration.
    \end{enumerate} 

%A simple mathematical model  that considered counterterrorism approaches was studied   by McCluskey and Santoprete in \cite{mccluskey2017bare}. 
 %In this paper we build on the model introduced in \cite{mccluskey2017bare} by adding a treatment compartment. This allows us to  consider de-radicalization in our  analysis. 
The development of viable intervention strategies to mitigate radicalization
and violence requires a thorough understanding of the  radicalization process, prevention, disengagement and 
deradicalization programs. Mathematical models can provide a first step in this direction. 
The aim of this paper is to  use a compartmental epidemiological model to analyze CVE programs, focusing on prevention and 
deradicalization initiatives.

The use of differential equations to describe social science problems dates back, at least,  to the  work of 
Lewis F. Richardson \cite{richardson1935mathematical} who pioneered the application of mathematical  techniques by  studying
the causes of war, and the relationship between arms race and the eruption of war. A summary of his
research was published posthumously in the book \cite{richardson1960arms}.

Modern applications of compartmental models to the social sciences range from models of political party growth,
to models of the spread of crime (see for instance \cite{crisosto2010community,hayward1999mathematical,jeffs2016activist,mcmillon2014modeling,mohammad2017analysis,
    romero2011epidemiological,sooknanan2013catching,sooknanan2016modified,sooknanan2018mathematical}). 
In recent years compartmental models have also been used to study  terrorism, the spread of fanatic behavior,
and radicalization \cite{castillo2003models,camacho2013development,galam2016modeling,mccluskey2018bare,santoprete2018global,nathanmodelling}. 
Furthermore, an age-structure model of radicalization was considered by Chuang, Chou and D'Orsogna 
\cite{chuang2018age}, a bi-stable model of radicalization within sectarian conflict was studied Chuang, D'Orsogna and Chou \cite{chuang2018bistable}, and a game theoretic model of radicalization was analyzed by Short, McCalla and d'Orsogna \cite{short2017modelling}.

The model we study here  extends the   one considered  in  \cite{santoprete2018global} 
by including a vaccinated class.  The purpose of this model is to  analyze two of the  CVE strategies,
namely prevention programs  and de-radicalization programs. 
As in \cite{mccluskey2018bare,santoprete2018global} we use Lyapunov functions to study the global
stability of the equilibria of the model and   the basic reproduction number $ \mathcal{R} _0 $ to
assess initiatives for combating terrorism.

 %We divide the population into four compartments, $ (S) $ susceptible, $ (E) $ extremists, $ (R) $ recruiters, and $ (T) $ treatment (see Figure \ref{fig:1}).

Although the literature indicates some degree of success of CVE programs, 
according to \cite{mastroe2016surveying}, there is little consensus regarding the validity of
CVE prevention programs or disengagement/de-radicalization programs, largely due to the lack of empirical data.
Furthermore, it is very difficult to evaluate these programs since  indicators of success and measures
of efficacy remain elusive \cite{ohalloran2017challenges}.
These are  key issues, since the  degree of government support for these programs  depends, to a large extent, 
on  demonstrating their effectiveness. 
The results we present in this paper are theoretical in nature and  are fairly  independent from the specific
choices of the parameter values.  Our model can, in principle, be used to evaluate the efficacy of
CVE programs in combating terrorism whenever empirical data are known.

\section{Equations}

%%%%%%%%%%%%%%%%%%%%%%%%%%%%%%%%%%%%%%%%%%%%%%%%%%
% Define block styles
%\tikzstyle{decision} = [diamond, draw, fill=blue!20, 
%    text width=4.5em, text badly centered, node distance=3cm, inner sep=0pt]

\tikzstyle{block} = [rectangle, draw, fill=blue!20, 
    text width=5em, text centered, rounded corners, minimum height=4em]
\tikzstyle{line} = [draw, -latex]
%\tikzstyle{cloud} = [draw, ellipse,fill=red!20, node distance=3cm,
%    minimum height=2em]
%%%%%%%%%%%%%%%%%%%%%%%%%%%%%%%%%%%%%%%%%%%%%%%%%
We use a compartmental model to describe the dynamics. We divide the population at risk 
of adopting an extreme ideology into five compartments 
%We model the spread of extreme ideology as a contact process.
%We assume that within the full population there is a subpopulation potentially at risk of
%adopting the ideology. We partition this   subpopulation  into five compartments:
\begin{enumerate}
    \item $(S)$  Susceptible
    \item  $(E)$ Extremists
    \item $ (R) $ Recruiters
    \item $ (T) $ Treated
    \item $ (V) $ Vaccinated. 
\end{enumerate}
Our model is extends the treatment model studied in \cite{santoprete2018global} by adding a     vaccinated compartment $ (V) $. This allows us  to  describe individuals in prevention programs. Our transfer diagram is similar, but different, to the one proposed by Yang, et al.  \cite{yang2016global} to model  the spread of tuberculosis with vaccination and treatment. The differences in the models are enough to create some complications in the construction of Lyapunov functions for our problem. 
The transfer diagram for our system is given below.

\begin{figure}[h]
\begin{tikzpicture}
\tikzset{VertexStyle/.style={circle,fill=blue!10,draw,minimum size=1cm,inner sep=5pt},
            }
    
  %\tikzset{VertexStyle/.style = {shape=circle, fill=white, minimum size = 20pt}} 
   %\Vertex[Math,L=^{238}\mathrm{U\,},x=0 ,y=7]{U238}
   \Vertex[Math,L=S,x=0,y=0]{S}
   \Vertex[Math,L=E,x=4,y=2.5]{E} 
   \Vertex[Math,L=R,x=4,y=-2.5]{R}
   \Vertex[Math,L=T,x=8,y=0]{T}
   \Vertex[Math,L=V,x=0,y=2.5]{V}
   \Vertex[empty,x=-2.5,y=0]{F0} 
   \Vertex[empty,x=0,y=-2.5]{F1}
   \Vertex[empty,x=4,y=4.5]{F2}
   \Vertex[empty,x=4,y=-4.5]{F3}
   \Vertex[empty,x=8,y=-2.5]{F4}
   \Vertex[empty,x=0,y=4.5]{F5}
   \Vertex[empty,x=-2.5,y=2.5]{F6}     
  \tikzset{EdgeStyle/.append style={->}} 
   \Edge[label = $q_E\beta SR$](S)(E)
   \Edge[label = $q_R\beta SR$](S)(R)
   \Edge[label = $c_R R$,style={bend left}](R)(E)
   \Edge[label = $c_E E$, style={bend left}](E)(R)  
   \Edge[label = $p_R R$](R)(T)
   \Edge[label = $p_E E$](E)(T) 
   \Edge[label = $(1-k)\delta T$, style={bend right}](T)(E)  
   \Edge[label = $p_S\Lambda$](F0)(S)
   \Edge[label = $\mu S$](S)(F1)
   \Edge[label = $(\mu+d_E) E$](E)(F2)
   \Edge[label = $(\mu+d_R) R$](R)(F3)  
   \Edge[label = $(\mu+k\delta) T$](T)(F4) 
   \Edge[label = $p_V\Lambda$](F6)(V)
   \Edge[label = $c_V V$](V)(S)
   \Edge[label = $\sigma q_E \beta VR$](V)(E)
   \Edge[label = $\mu V$](V)(F5) 
    \end{tikzpicture}
\end{figure}

Recruitment occurs in the system with rate constant $\Lambda > 0$. Of these individuals,  a fraction 
$ p _V $ enters the vaccination compartment, while a fraction $ p _S = 1 - p _V $  enters  the susceptible population.
The rate at which susceptibles are recruited is  $\beta S R$. A fraction $ q _E $ of the newly  recruited individuals
are assumed to transfer to  the extremist class, while the remainder $q_R = 1 - q_E \ll 1 $  transfer to   the recruiter class. 
Vaccinated individuals  are recruited at a reduced rate $\sigma q _E  \beta V R $, with $ 0 \leq \sigma \leq 1 $. 
The rate constant at which an individual leaves the extremist compartment to became a recruiter is $ c _E $,
while $ c _R $ is the rate constant at which a recruiter abandons the recruiter class to become an extremist. 
The natural death rate  constant is $ \mu $, $ d _E $ and $ d _R $ are supplementary death rate constants 
for individuals in compartments $ E $ and $ R $, respectively. The additional rates $ d _E $ and $ d _R $
take into consideration individuals that die or are sentenced to lifelong incarceration as a result of police or military action. 
The rate constants of extremists transferred to the treatment compartment is $ p _E $, while $ p _R $ is 
the rate constant of extremists moving to compartment $ T $. 

Treated individuals exit  $ T $ at a rate $ \delta$. 
We remove a  fraction $ k \in [0,1]$ of treated individuals, since   effectively treated individuals
are  de-radicalized forever. 
For a  fraction $ 1-k $ of treated individuals the de-radicalization program is unsuccesful.
These individuals get into the $ E $ compartment  after treatment. 

Based on the above assumptions we obtain the following model: 
\begin{equation}\label{eqn:modela}
 \begin{aligned}
 S' & = p _S \Lambda + c _V V  - \mu S - \beta SR		\\
 E' & = q_E\beta SR + \sigma  q _E \beta V R  - (\mu + d_E + c_E+ p _E ) E + c_R R +(1-k)\, \delta T		\\%+(1-k)r _E \,\delta T	
 R' & = q_R \beta SR + c_E E - (\mu + d_R + c_R+p _R ) R\\%+(1-k)r _R \,\delta T	\\
T' & = p _E E + p _R R - (\mu +\delta)T\\ 
V' & = p _V \Lambda - c _V V - \mu V - \sigma q _E \beta V R 
 \end{aligned}
\end{equation}
where $ q _E + q _R = 1 $, $ q _E , q _R \in [0,1] $. 
To simplify system   \eqref{eqn:modela} we introduce the following parameters $ b _E = \mu + d_E + c_E+ p _E $,
$ b _R =  \mu + d_R + c_R+ p _R  $, $ b _T = \mu + \delta $ and $ b _V = c _V + \mu $.
Using these new constants in \eqref{eqn:modela} yields:
\begin{equation}\label{eqn:modelb}
 \begin{aligned}
 S' & = p_S\Lambda  + c _V V - \mu S - \beta SR		\\
 E' & = q_E\beta(S + \sigma V) R  - b _E  E + c_R R +(1-k) \, \delta T		\\%+(1-k)r _E \, \delta T	
 %q_E\beta SR + \sigma q _E \beta V R  - b _E  E + c_R R +(1-k) \, \delta T		\\%+(1-k)r _E \, \delta T	
 R' & = q_R \beta SR + c_E E - b _R  R\\%(1-k)r _R \,\delta T\\
T' & = p _E E + p _R R - b _T T  \\
V' & = p _V \Lambda - b _V V - \sigma  q _E \beta V R. 
 \end{aligned}
\end{equation}
It is not difficult to show that the region
\[ \Delta= \left \{ (S, E , R,T,V) \in \mathbb{R}  ^5 _{ \geq 0 } : S +E+R+T+V \leq \frac{ \Lambda } { \mu } \right \}\]
is a compact positively invariant and attracting set that attracts all solutions of \eqref{eqn:modelb} with initial
conditions in $ \mathbb{R}  ^5 _{ \geq 0 } $. See Proposition 2.1 in \cite{santoprete2018global} for a proof 
of a similar statement. 

\section{Radicalization-free equilibrium and basic reproduction number $ \mathcal{R} _0 $}

There is a unique equilibrium with  $ E = R =T= 0 $  given by $ x _0 = \left( S _0 ,0,0,0,V _0  \right) $, where 
\begin{equation}\label{eqn:S0V0}
   S _0 = \frac{ \left( p _S + \frac{ c _V } { b _V } p _V  \right) \Lambda } { \mu }= \frac{ \left( p_S+ \frac{ c_V } { \mu }  \right) \Lambda } { b_V } \quad\quad V _0 = \frac{ p _V \Lambda } { b _V }   
\end{equation} 

We denote by $ \mathcal{R} _0 $ the spectral radius of  the matrix $ G $ evaluated at $ x _0 $.  
%The basic reproduction number $\mathcal{R}_0 $  is the spectral radius of the next generation matrix $G $
%calculated at $ x _0 $.
$ \mathcal{R} _0 $ is called the {\it basic reproduction number} and can be obtained  as outlined 
by Van Den Driessche and Watmough \cite{van2002reproduction}.
First we identify the infected classes, that in this example turn out to be  $ E,R,T $.

Suppose  $ \mathcal{F}  _E $, $ \mathcal{F}  _R $ and $ \mathcal{F}  _T $  are the rates of arrival of newly radicalized individuals
in the compartment $ E $, $ R $, and $ T $,  respectively. Let $ \mathcal{V}  _j = \mathcal{V} 
_j ^{ - } - \mathcal{V}  _j ^{ + } $,  with $ \mathcal{V}  _j ^{ + } $  be the rate of
transmission of individuals into compartment $ j \in \{E,R,T\}$    by all remining methods, and 
$ \mathcal{V}  _j ^{ - } $  the rate of removal of individuals from compartment $j$, where 
$ j $ is one of $ E, R $, and $ T $. In our problem  
\[
    \mathcal{F} = \begin{bmatrix}
        \mathcal{F}  _E   \\
        \mathcal{F}  _R\\
        \mathcal{F}  _T 
    \end{bmatrix} =\beta  
    \begin{bmatrix}
       q _E (SR+\sigma VR)\\
       q _R SR\\
       0
  \end{bmatrix}  
%\aand
\]
and 
\[
\mathcal{V} =  \begin{bmatrix}
          \mathcal{V}  _E   \\
          \mathcal{V}  _R\\ 
          \mathcal{V} _T  
    \end{bmatrix} = 
    \begin{bmatrix}
     b _E E - c _R R - (1 - k) \delta T    \\%- (1 - k) r _E \delta T%
     b _R R - c _E E\\  %- (1 - k) r _R \delta T 
     b _T T - (p _E E + p _R R)    
  \end{bmatrix}.
\]
Then consider  the matrices 
\[
    F = \begin{bmatrix}
        \frac{ \partial \mathcal{F}  _E } { \partial E }  & \frac{ \partial \mathcal{F}  _E } { \partial R }  &  \frac{ \partial \mathcal{F}  _E } { \partial T }\\[1em]
        \frac{ \partial \mathcal{F}  _R } { \partial E }  & \frac{ \partial \mathcal{F}  _R } { \partial R } &  \frac{ \partial \mathcal{F}  _R } { \partial T }\\[1em]
        \frac{ \partial \mathcal{F}  _T } { \partial E }  & \frac{ \partial \mathcal{F}  _T } { \partial R } &  \frac{ \partial \mathcal{F}  _T } { \partial T }\\

    \end{bmatrix} (x_0)
  \aand
  V = \begin{bmatrix}
        \frac{ \partial \mathcal{V}  _E } { \partial E }  & \frac{ \partial \mathcal{V}  _E } { \partial R }  &  \frac{ \partial \mathcal{V}  _E } { \partial T }  \\[1em]
        \frac{ \partial \mathcal{V}  _R } { \partial E }  & \frac{ \partial \mathcal{V}  _R } { \partial R } &  \frac{ \partial \mathcal{V}  _R } { \partial T }\\[1em]
         \frac{ \partial \mathcal{V}  _T } { \partial E }  & \frac{ \partial \mathcal{V}  _T } { \partial R } &  \frac{ \partial \mathcal{V}  _T } { \partial T }
    \end{bmatrix} (x_0),
\]
which in our problem take the form 
\[
    F =\beta  \begin{bmatrix}
         0&   q _E(S_0+\sigma V _0 ) & 0 \\
         0 & q _R S_0 & 0\\
         0 & 0 & 0
    \end{bmatrix}
\aand
%\]
%and 
%\[
   V =  \begin{bmatrix}
       b _E & - c _R & - \alpha _E  \\
       - c _E & b _R  & 0\\
       - p _E & - p _R & b _T 
  \end{bmatrix},
\]
where $ \alpha _E = (1 - k) \delta $.  

Finally, we can compute the next generation matrix  $ G  = F V ^{ - 1 }$:  
%\begin{footnotesize} 
\begin{align*} 
    G & =-\frac{ \beta 
           }{b _T D}     
 \begin{bmatrix}
        0  & q _E(S_0+\sigma V _0 ) & 0  \\
        0& q _R S _0& 0 \\
        0 & 0 & 0 
    \end{bmatrix}
 \begin{bmatrix}
       - b _R b _T  & -(\alpha _E p _R + c _R b _T )& \alpha _E b _R  \\
      -c _E b _T  & \alpha _E p _E - b _E b _T & - \alpha _E c _E \\
     - b _R p _E - c _E p _R & - b _E p _R - c _R p _E & - b _E b _R + c _R c _E   
  \end{bmatrix}\\
 & = -\frac{ \beta } {b _T  D } \begin{bmatrix}
    -q _E c _E b _T(S _0 + \sigma V _0 )  & q _E(\alpha _E p _E - b _E b _T )(S _0 + \sigma V _0)    & - q _E \alpha _E c _E S _0  \\
   -q _R c _E b _T S _0   & q _R( \alpha_E p _E - b _E b _T) S _0  & - q _R \alpha _E c _E S _0\\
  0 & 0 & 0 
   \end{bmatrix},  
   \end{align*} 
%\end{footnotesize}
where $D = b _E b _R - c _E c _R   -\frac{ \alpha _E}{b _T } (b _R p _E + c _E p _R)>0 $.
Since the matrix $ G $  has  only one non-zero eigenvalue,  its spectral radius is:
\begin{equation}				\label{Rzero}
     \mathcal{R} _0 =  \frac {\beta S _0  (c _E q _E + b _E q _R - \frac{ \alpha _E p _E } { b _T } q _R )+
         \beta \sigma V _0 q _E c _E } { b _E b _R - c _E c _R - \frac{ \alpha _E } { b _T } (c _E p _R + b _R p _E) }.    
\end{equation}
%Note that the basic reproduction number $ \mathcal{R} _0 $ does not depend on $ \beta _2 $. This  suggests considering the special case $ \beta _2 = 0 $. 

\section{Global Asymptotic Stability of  $ x _0 $ for $ \mathcal{R} _0 < 1 $}
In this section we prove   the global asymptotic stability of  the equilibrium $ x _0 $.
For this purpose we introduce  the following  Lyapunov function 
   \[
       U = A\, 
       \frac{ (S- S_0)^2}{2S_0} +
       \frac{ (V - V _0) ^2 } { 2 V _0 }   + \frac{ H } { b _T c _E }, 
   \]
where  $ H $ is the function defined by 
\[ H(E,R,T) =  b _T c _E E + (b _T b _E - \alpha _E p _E) R + \alpha _E c _E T. \]
Let 
\[Q= b _T b _E - \alpha _E p _E=(\mu + \delta) (\mu + d _E + c _E) + \mu p _E + k \delta p _E >0,  \]
then 
\[ A = \beta \frac{ c _E q _E + b _E q _R - \frac{\alpha _E p _E}{b _T } q _R }{c _E } \] 
is a positive constant since 
\[ c _E q _E + \frac{ q _R}{b _T }  (b _E b _T - \alpha _E p _E) = c _E q _E + \frac{ q _R}{ b _T }  Q>0. \]

We can now prove the following theorem 

\begin{theorem}\label{thm:x0}
   Suppose $ A p _V < 4 $, and  $ \mathcal{R} _0 \leq 1 $ then $ x _0 $ is globally asymptotically stable on  $ \mathbb{R}  ^4 _{ \geq 0 }$.  
\end{theorem} 

\begin{proof}
%%%%%%%%%%%%%%%%%%%%%%%%%%%%%%%%%%%%%%%%%%%%%%%
We study the stability of $ x _0 $ by taking the  Lyapunov function   
   \[
       U = A\, 
       \frac{ (S- S_0)^2}{2S_0} +
       \frac{ (V - V _0) ^2 } { 2 V _0 }   + \frac{ H } { b _T c _E } 
   \]
where  $ H $ is the function defined by 
\[ H(E,R,T) =  b _T c _E E + (b _T b _E - \alpha _E p _E) R + \alpha _E c _E T, \]
with 
\[Q= b _T b _E - \alpha _E p _E=(\mu + \delta) (\mu + d _E + c _E) + \mu p _E + k \delta p _E >0,  \]
and 
\[ A = \beta \frac{ c _E q _E + b _E q _R - \frac{\alpha _E p _E}{b _T } q _R }{c _E } \] 
is a positive constant since 
\[ c _E q _E + \frac{ q _R}{b _T }  (b _E b _T - \alpha _E p _E) = c _E q _E + \frac{ q _R}{ b _T }  Q>0. \]
%\fi
%%%%%%%%%%%%%%%%%%%%%%%%%%%%%%%%
Using  \eqref{eqn:S0V0} we obtain
\begin{equation}\label{eqn:Lambda} 
\begin{aligned}
   p_V\Lambda  & = b _V V _0 \\
   p _S\Lambda & = \mu S _0 - c _V V _0. 
\end{aligned} 
\end{equation}
Differentiating $ U $ with respect to $ t $ along the  trajectories of the system \eqref{eqn:modelb}, and using both equations in \eqref{eqn:Lambda}  to rewrite $ p _V \Lambda $ and $ p _S \Lambda  $  yields 
\begin{align*}
    U'  =&  A \frac{ (S - S _0 )} { S _0 }  S' + \frac{ (V - V _0) } { V _0 }   V' 
    + \frac{ b _T c _E E' + (b _T b _E - \alpha _E p _E) R' + \alpha _E c _E T' } { b _T c _E }
    \\
    =&  f (S, V) - A \beta \frac{ (S - S _0) ^2 } { S _0 } R - \sigma \beta q _E \frac{ (V - V _0) ^2 } { V _0 } R   - A \beta (S - S _0) R \\
    & - \sigma q _E \beta (V - V _0) R  + A \beta S R + \sigma q _E \beta V R + \frac{ D  } {  b _T c _E } \, R   \\
    \leq &  f (S, V)   - A \beta (S - S _0) R  - \sigma q _E \beta (V - V _0) R  + A \beta S R + \sigma q _E \beta V R + \frac{ D  } {  b _T c _E } \, R   \\
    = & f (S, V)  + \left[A \beta S_0  + \sigma q _E \beta V_0  + \frac{ D  } { b _T c _E } \right]\, R \\
    = & f (S, V)  + \frac{ D } { b _T c _E } \left[1+ b _T  c _E \frac{ A \beta S_0  + \sigma q _E \beta V_0}{D }  \right]\, R\\
    = &  f (S, V)  + \frac{ D } { b _T c _E } \left[1 
   - \mathcal{R} _0  \right]\, R 
    \end{align*} 
where $ f (S, V) =  - A\, \mu \frac{ (S - S _0) ^2 } { S_0 } - b_V \frac{ (V - V _0) ^2 } { V_0 } + A\, c _V \frac{ (S - S _0) (V - V _0) } { S_0 }  $.
%where $ D = b _E b _R - c _E c _R - \frac{ \alpha _E } { b _T } (p _E b _R + c _E p _R) $.
It remains to show that $ f (S, V) \leq 0 $ and $ f (S, V) = 0 $ if and only if $ S = S _0 $ and $ V = V _0 $.

Note that the Hessian matrix  of $ f (S, V) $ is 
\[
    \begin{bmatrix}
        -\frac{2 A \mu }{S_0} & \frac{ A c_V } { S _0 }  \\[10pt]
        \frac{A c _V } { S _0 }  &  - \frac{ 2 b _V } { V _0 } 
    \end{bmatrix} 
\]
and its determinant is $ - \frac{ A(A V _0 c _V ^2 - 4 S _0 b _V \mu) } { S _0 ^2 V _0 } $. By the second derivative test, $ (S,V) = (S_0,V_0) $ is  a maximum when the  determinant is  positive, and  $-\frac{2 A \mu }{S_0} <0 $. The latter inequality holds, hence,  it only remains to study the determinant. 
Since $ A >0 $ we must have $ -A V _0 c _V ^2 + 4 S _0 b _V \mu>0 $. Substituting $ S _0 $ and $ V _0 $ in this expression we obtain 
\[
    -A V _0 c _V ^2 + 4 S _0 b _V \mu=\frac{ \Lambda } { b _V } \left[ - A p _V c _V ^2 + 4 b _V c _V + 4 \mu p _S b _V \right] 
\]
Since $ b _V = c _V + \mu $  the expression inside the bracket reduces to 
\[
    (4-A p _V ) c _V ^2 + 4 \mu c _V + 4 \mu p _S p _V,
\]
which proves that the determinant is always positive if $ p _V A < 4 $. 

In this case, for $ \mathcal{R} _0 \leq 1 $, we have $ U' \leq 0 $, with equality if and only if
$ S = S _0 $ and $ V = V _0 $, and $ R = 0 $. The largest invariant set for which $ U' = 0 $, then consists of just the equilibrium $ x _0 $. 
The theorem then follows from  LaSalle's Invariance Principle. 
\end{proof} 
Note that, Theorem \ref{thm:x0} holds when the parameters satisfy the inequality $ A p _V < 4 $. This does not seem to pose a substantial restriction in the allowable value of the constants, since $ A p _V $  is a small number for any reasonable choice. For instance if we choose the parameters as in Figure 3, then $ A p _V = 6.633301503 \times 10^{-7} $. It may be of interest to see if it is possible to remove the restriction $ A p _V < 4 $ by using a different Lyapunov function. 

\section{Endemic Equilibrium}
We now look for equilibria of  \eqref{eqn:modelb} for which at least one of the populations  $ E ^\ast , R ^\ast, T ^\ast  $ and $ V ^\ast $ 
is different from zero. We call such point an {\it endemic equilibrium} and we denote it by  
$ x ^\ast = (S ^\ast , E ^\ast , R ^\ast, T ^\ast, V ^\ast  ) \in \mathbb{R} ^5_{>0} $.
The  endemic equilibria of \eqref{eqn:modelb} are given by the following system of equations
\begin{equation}\label{eqn:system}
 \begin{aligned}
 &  p_S\Lambda  + c _V V - \mu S - \beta SR	=0	\\
  &  q_E\beta(S + \sigma V) R  - b _E  E + c_R R +(1-k) \, \delta T=0		\\
  &  q_R \beta SR + c_E E - b _R  R=0\\
 &  p _E E + p _R R - b _T T  =0\\
 &  p _V \Lambda - b _V V - \sigma  q _E \beta V R =0.
 \end{aligned}
\end{equation}
Solving  the first, third, fourth and fifth equation in \eqref{eqn:system} and treating $ R^* $ as a parameter  we obtain the following
\begin{equation}\label{eqn:solution}
    \begin{aligned}
       S ^\ast & = \frac{ \Lambda } { \mu + \beta R ^\ast } \left( p _S + \frac{ c _V p _V } { b _V + \sigma q _E \beta R ^\ast } \right) \\
       E ^\ast & = \omega R ^\ast \\
       T ^\ast & = \frac{ p _E \omega + p _R } { b _T } R ^\ast \\
       V ^\ast & = \frac{ p _V \Lambda } { b _V + \sigma q _E \beta R ^\ast },
    \end{aligned} 
\end{equation} 
where 
\[\omega = \frac{ b _R - q _R \beta S ^\ast } { c _E } = 
    \frac{b _R b _T q _E + b _T c _R q _R + p _R q _R \alpha _E  } { q _R (b _E b _T - \alpha _E p _E) + c _E q _E b _T  }>0,\]
since $ (b _E b _T - \alpha _E p _E) >0 $.
%Substituting the first expression above in the expressions for $ E ^\ast $ and $ T ^\ast $, yields $ S ^\ast , E ^\ast , T ^\ast $  and $ V ^\ast $ as a function of $ R ^\ast $. 
Substituting the expressions in \ref{eqn:solution}
  in the second equation of \eqref{eqn:system} yields the following equation for $ R ^\ast $ 
\[
     R ^\ast \left( \alpha  _2 (R ^\ast ) ^2   + \alpha  _1 R ^\ast + \alpha  _0 \right)  = 0 
 \]
 where
 \begin{align*}
     \alpha  _2 & = - \beta ^2 q _E \sigma D b _T \\
     \alpha  _1 & = - \beta b _T (b _V + \mu \sigma q _E) D + \sigma (((b _E q _R + c _E q _E )b _T +  (k - 1) \delta  q _R p _E p_S + b _T c _E p _V )) \beta ^2 q _E \Lambda  \\ 
     \alpha  _0& = b_T b_V \mu D \left( \mathcal{R} _0 - 1 \right)
 \end{align*}  
This equation has solution if either $ R ^\ast = 0 $ or 
\begin{equation}\label{eqn:quadratic}
    \Phi (R ^\ast) =  \alpha  _2 (R ^\ast ) ^2   + \alpha  _1 R ^\ast + \alpha   _0 =0.
\end{equation}
The  case $ R ^\ast = 0 $  must be excluded since   it yields a solution with $ E ^\ast = R ^\ast = T ^\ast=0 $,
which was already known.   Consider equation \eqref{eqn:quadratic}. Clearly, $ \alpha _2<0 $  (since $ D >0 $) and $\alpha _0 >0 $ whenever $ \mathcal{R} _0 > 1 $. It follows that   \eqref{eqn:quadratic} has a unique positive root $ R ^\ast $ if  $ \mathcal{R} _0 > 1 $. 
Since $ \Delta $ is attracting within $ \mathbb{R}  ^5 _{ \geq 0} $ we have that $ R 
^\ast \in (0 , \Lambda / \mu) $. 

%By Proposition \ref{prop:invariant_region} we have that $ R 
%^\ast \in (0 , \Lambda / \mu) $. 

Note that we can also find  an useful  formula for  $ S ^\ast $ as a function of $ V ^\ast $.  This formula can be obtained by substituting the equations for $ E ^\ast $ and $ T ^\ast $ given in \eqref{eqn:solution} into the second line of \eqref{eqn:system}, and solving for $ S ^\ast $. This gives
\begin{equation} \label{eqn:S_ast}
    S ^\ast = \frac{ b _E b _R - c _E c _R - \frac{ \alpha _E } { b _T } (p _R c _E + p _E b _R) - q _E \beta \sigma c _E V ^\ast } { \beta (c_Eq _E + b _E q _R - \frac{ \alpha _E p _E } { b _T } q _R )} . 
\end{equation} 

\begin{theorem}\label{thm:x_ast} The endemic equilibrium $ x ^\ast $ of \eqref{eqn:modelb} is globally
    asymptotically stable on $ \mathbb{R}   _{ >0 } ^5 $, whenever $ \mathcal{R} _0 >1 $.
\end{theorem} 
\begin{proof}
    Consider the Lyapunov function
    \[
        W = S ^\ast g \left( \frac{ S } { S ^\ast } \right) + a _1 E ^\ast g \left( \frac{ E } { E ^\ast } \right) + a _2 R ^\ast g \left( \frac{ R } { R ^\ast } \right) + a _3 T ^\ast g \left( \frac{ T } { T ^\ast } \right) + a _4 V ^\ast g \left( \frac{ V } { V ^\ast } \right) 
    \]
with $ g (x) = x - 1 - \ln{x} $, and $ a _i >0 $ ($ i = 1, \ldots , 4 $), where the $ a _i $s  are constants to be determined. 
Clearly, $ W $ is $ C ^1 $, $ W (x ^\ast) = 0 $ , and $ W >0 $ for any $ p \in \mathbb{R}  ^5 _{ >0 } $ such that $ p \neq  x ^\ast $. 
Computing the derivative of  $ W $ along the  solutions of \eqref{eqn:modelb} gives 
\begin{align*} 
    W'  =&  \left( 1 - \frac{ S ^\ast } { S } \right) S' + a _1 \left( 1 - \frac{ E ^\ast } { E } \right) E' + a _2 \left( 1 - \frac{ R ^\ast } { R } \right) R' + a _3 \left( 1 - \frac{ T ^\ast } { T } \right) T' + a _4 \left( 1 - \frac{ V ^\ast } { V } \right) V '\\
     =&  \left( 1 - \frac{ S ^\ast } { S } \right) [p _S \Lambda + c _V V - \mu S - \beta SR]\\
     & + a _1 \left( 1 - \frac{ E ^\ast } { E } \right) [q _E \beta (S + \sigma V) R - b _E E + c _R R + (1 - k) \delta T  ]\\
    & +  a _2 \left( 1 - \frac{ R ^\ast } { R } \right)[q _R \beta S R + c _E E - b _R R] + a _3 \left( 1 - \frac{ T ^\ast } { T } \right) [p _E E + p _R R - b _T T]\\
   & +    a _4 \left( 1 - \frac{ V ^\ast } { V } \right) [p _V \Lambda   - b _V V - \sigma  q _E \beta V R]\\
   = & C 
  - (\mu + a _2 \beta q _R R ^\ast) S + (a _1 q _E + a _2 q _R - 1) \beta SR + (- a _1 b _E  + a _2 c _E  + a _3 p _E ) E \\
  & + (S ^\ast \beta + a _1 c _R - a _2 b _R  + a _3 p _R  + a _4 \sigma q _E \beta V ^\ast  ) R  +(a _1 \alpha _E - a _3 b _T)T  +(c _V - a _4 b _V )V\\
 &  + q _E \beta \sigma (a _1  - a _4  )VR- p _S \Lambda \frac{ S ^\ast } { S } - a _3 p _E T ^\ast \frac{ E } { T } -a _2 c _E R ^\ast \frac{ E } { R } - a _3 p _R T ^\ast \frac{ R } { T } - a _1 \alpha _E E ^\ast \frac{ T } { E }\\
&  - a _1 c _R E ^\ast \frac{ R } { E } - a _1  \beta q _E E ^\ast \frac{ SR } { E } - a _1 \sigma q _E \beta E ^\ast \frac{ V R } { E } - a _4 p _V \Lambda \frac{ V ^\ast } { V } - c _V V \frac{ S ^\ast } { S }. 
\end{align*} 
where $ C = \Lambda + \mu S ^\ast + a _1  b _E E ^\ast + a _2 b _R R ^\ast + a _3  b _T T ^\ast + a _4 b _V V ^\ast $.
 For simplicity denote $ w = \frac{ S } { S ^\ast } $, $ x = \frac{ E } { E ^\ast } $, $ y = \frac{ R } { R ^\ast } $, $ z = \frac{ T } { T ^\ast } $, and $ v = \frac{ V } { V ^\ast } $.
Then
\begin{align*}
   W' = & C 
  - (\mu + a _2 \beta q _R R ^\ast) S ^\ast w + (a _1 q _E + a _2 q _R - 1) \beta S ^\ast R ^\ast wy  + (- a _1 b _E  + a _2 c _E  + a _3 p _E ) E ^\ast x  \\
  & + (S ^\ast \beta + a _1 c _R - a _2 b _R  + a _3 p _R  + a _4 \sigma q _E \beta V ^\ast  ) R ^\ast y  +(a _1 \alpha _E - a _3 b _T)T ^\ast z   +(c _V - a _4 b _V )V ^\ast v\\
  & + q _E \beta \sigma (a _1  - a _4  )V ^\ast R ^\ast v y- p _S \Lambda \frac{ 1} { w } -a _3  p _E 
 E ^\ast \frac{ x } { z }  -a _2 c _E E ^\ast \frac{ x } { y }  - a _3 p _R R^\ast \frac{ y } { z } - a _1 \alpha _E T ^\ast \frac{ z } { x }  \\
 & - a _1 c _R R ^\ast \frac{ y } { x } - a _1  \beta q _E S ^\ast R ^\ast  \frac{ wy } { x } - a _1 \sigma q _E \beta V ^\ast R ^\ast \frac{ vy } {x}  - a _4 p _V \Lambda \frac{ 1 } { v } - c _V V ^\ast \frac{ v } { w }. 
\end{align*} 
Following the method used in \cite{li2012algebraic} and \cite{lamichhane2015global} we  introduce the   following set
\[
    \mathcal{D} = \left \{v, w,x,y,z,vy, wy, \frac{ 1 } { v } ,  \frac{ 1 } { w },  \frac{ v } { w } ,\frac{ x } { z } , \frac{ x } { y } , \frac{ y } { z } , \frac{ z } { x } , \frac{ y } { x } , \frac{ v y } { x } ,  \frac{ wy } { x } \right \}. 
\]
We now list all the subsets of $ \mathcal{D} $ for which the product of all functions within each subset is equal to one
%There are at most nine subsets of $ \mathcal{D} $ such that the product of all functions within the subset
%is equal to one, they are
\begin{multline}\label{eqn:subsets}
    \left \{ v, \frac{ 1 } { v } \right \}, \left \{ w, \frac{ 1 } { w } \right \}, \left \{ \frac{ x } { y } , \frac{ y } { x } \right \}, \left \{ \frac{ x } { z } , \frac{ z } { x } \right \},\\ \left \{ \frac{ 1 } { v } , \frac{ vy } { x } , \frac{ x } { y } \right \}, \left \{ \frac{ 1 } { w } , \frac{ wy } { x } , \frac{ x } { y } \right \} ,
    \left \{ \frac{ z } { x } , \frac{ y } {z } , \frac{ x } { y } \right \}, \left \{ w,  \frac{ 1 } { v } , \frac{ v } { w }  \right \}, \left \{ \frac{ 1 } { v }, \frac{ v } { w }  , \frac{ x } { y } , \frac{ wy } { x } \right \} .   
\end{multline}
We associate the following terms to the subsets in equation \eqref{eqn:subsets}:

%To these  subsets of variables we associate the following terms 
\begin{align*} 
     & \left( 2 - v - \frac{ 1 } { v } \right), \left( 2 - w - \frac{ 1 } { w } \right), \left( 2 - \frac{ x } { y } - \frac{ y } { x } \right), \left( 2 - \frac{ x } { z } - \frac{ z } { x } \right),\\
   & \left( 3 - \frac{ 1 } { v } - \frac{ vy } { x } - \frac{ x } { y } \right), \left( 3 - \frac{ 1 } { w } - \frac{ wy } { x } - \frac{ x } { y } \right) , \left( 3 - \frac{ z } { x } - \frac{ y } { z } - \frac{ x } { y } \right), \left( 3 -w- \frac{ 1 } { v }  - \frac{ v } { w } \right), \\
   & \left(4 - \frac{ 1 } { v } - \frac{ v } { w } - \frac{ x } { y } - \frac{ wy } { x }  \right)  .   
\end{align*} 
As in  \cite{li2012algebraic}, we define a Lyapunov function using the terms above:
\begin{multline}
   H (v, w,x, y ,z) =   b _1  \left( 2 - v - \frac{ 1 } { v } \right) + b _2  \left( 2 - w - \frac{ 1 } { w } \right)
   + b _3 \left( 2 - \frac{ x } { y } - \frac{ y } { x } \right)\\
   + b _4  \left( 2 - \frac{ x } { z } - \frac{ z } { x } \right)
   + b _5 \left( 3 - \frac{ 1 } { v } - \frac{ vy } { x } - \frac{ x } { y } \right)\\
  + b _6  \left( 3 - \frac{ 1 } { w } - \frac{ wy } { x } - \frac{ x } { y } \right) 
 + b _7  \left( 3 - \frac{ z } { x } - \frac{ y } { z } - \frac{ x } { y } \right)\\
 +  b _8 \left( 3 -w- \frac{ 1 } { v }  - \frac{ v } { w } \right),+ b _9 \left(4 - \frac{ 1 } { v } - \frac{ v } { w } - \frac{ x } { y } - \frac{ wy } { x }  \right) 
\end{multline} 
where  $ b _1 , \ldots , b _9 $ are unknown constants. 
We look for solutions of the equation $ G (v, x,y,z) = H (v , w, x, y , z) $ with  $ a _i >0 $ ($ i = 1, \ldots , 4 $)  and $ b _k \geq 0 $ 
($k = 1, \ldots , 9 $)

 Setting  like terms in $ G $ and $ H $ equal  yields:
\begin{equation}\label{eqn:monster} 
\begin{aligned}
  vy: & \quad q _E \beta \sigma (a _1 - a _4) = 0 \\ 
   wy: & \quad a _1 q _E + a _2 q _R - 1 = 0\\
    x: & \quad - a _1 b _E + a _2 c _E + a _3 p _E = 0 \\
    y: & \quad S ^\ast \beta + a _1 c _R - a _2 b _R + a _3 p _R + a _4 \sigma q _E \beta V ^\ast=0 \\
    z : & \quad a _1 \alpha _E - a _3 b _T = 0 \\[10pt]
   v:& \quad b _1 = - (c _V - a _4 b _V) V ^\ast \\
   y x ^{ - 1 } :& \quad b _3 = a _1 c _R R ^\ast \\
   x z ^{ - 1 } :&  \quad b _4 = a _3 p _E E ^\ast \\
   vy x ^{ - 1 } :& \quad  b _5 = a _1 \sigma q _E \beta V ^\ast R ^\ast \\
   y z ^{ - 1 } :& \quad b _7 = a _3 p _R R ^\ast \\
   z x ^{ - 1 } :& \quad b _4 + b _7 = a _1 \alpha _E T ^\ast \\[10pt]
   v ^{ - 1 } :&  \quad b _1 + b _5+b_8+b_9 = a _4 p _V \Lambda \\
   x y ^{ - 1 } :& \quad b _3 + b _5 + b _6 + b _7+b_9 = a _2 c _E E ^\ast \\[10pt]
   v ^0:&\quad 2 (b _1 + b _2 + b _3 + b _4) + 3 (b _5 + b _6 + b _7+ b _8 )+4 b _9 = C\\
   w:& \quad b _2+b_8 = (\mu + a _2 \beta q _R R ^\ast) S ^\ast \\  
   w ^{ - 1 } :&  \quad b _2 + b _6 = p _S \Lambda \\
    c w ^{ - 1 } : & \quad b _8 + b _9 =  c _V V ^\ast\\
    wy x ^{ - 1 } :& \quad b _6+b_9 = a _1 \beta q _E S ^\ast R ^\ast \\
\end{aligned}
\end{equation} 
The first five lines of \eqref{eqn:monster} form an overdetermined  consistent system of five equations in the variables $ a _1 ,\ldots , a _4 $, with solution 
\begin{align*}
    a _1 & = a _4 =  \frac{c_E}{ c _E q _E+ b _E q _R - \frac{ p _E } { b _T } \alpha _E q _R}>0\\ 
     a _2 & = \frac{ 1 } { q _R } -  \frac{ \frac{q_E}{q_R}c _E  }{ c _E q _E+ b _E q _R - \frac{ p _E } { b _T } \alpha _E q _R}>0\\ 
     a _3 & = \frac{\frac{c_E\alpha_E}{b_T}}{ c _E q _E+ b _E q _R - \frac{ p _E } { b _T } \alpha _E q _R}>0.
\end{align*}
The following six lines of \eqref{eqn:monster}  form  an overdetermined  consistent system of six equations in the variables $ b _1 ,b_3 , b _4,b_5,b_7 $, with a unique solution where $ b _1 ,b_3 , b _4,b_5,b_7 >0 $. The next two lines can be easily seen to be consistent with the other equations,  but because they are dependent from other equations they can be discarded. 

The last group of equations in \eqref{eqn:monster}  is a system of  five equations in the variables $ b _2, b_6, b_8$ and $ b_9 $.
It remains to prove that this system has a solution with $   b _2, b_6, b_8,  b_9>0 $. The augmented matrix of the system given by the last five equations of \eqref{eqn:monster} is 
\[ \left[\begin{array}{rrrr|c}
    %\begin{bmatrix}
       2 & 3 & 3 & 4 & d _1  \\
       1 & 0 &  1 & 0 & d _2 \\
       1 & 1 & 0 & 0 &  d_3\\
       0 & 0 & 1 & 1 &  d _4  \\
       0 & 1 & 0 & 1 & d _5 
 % \end{bmatrix}  
 \end{array}\right]    
\]
where $ d_1= C -2 (b _1 + b _3 + b _4 )-3 (b _5 + b _7)   $, $ d _2 = (\mu + a _2 \beta q _R R ^\ast)S ^\ast  $, $ d _3 = p _S \Lambda $, $ d _4 = c _V V ^\ast $, and $ d _5 = a _1 \beta q _E S ^\ast R ^\ast $. Performing row operations we can reduce the augmented matrix to row echelon form  
\begin{align*}
    \left[\begin{array}{rrrr|c}
    %\begin{bmatrix}
       1 & 0 & 1 & 0 & d _2  \\
      0 & 1 & - 1 & 0 & d _3 - d _2 \\
       0 & 0 & 1 & 1 &  d_4\\
       0 & 0 & 0 & 0 & \frac{1}{4} (d _1 + d _2 - 3 d _3) - d _4  \\
       0 & 0 & 0 & 0 & d _5 - d _3 + d _2 - d _4  
 % \end{bmatrix}  
 \end{array}\right]
\end{align*}
A computations involving the last two rows of the augmented matrix above shows that the system is consistent, but since only the first three rows are independent it follows that the solution is not unique. Solutions of the system are given by
\begin{align*} 
b _2 & = d _2 - t\\
b _6 & = (d _3 - d _2) + t\\
b _9 & = d _4 - t
\end{align*} 
with $ b _8 = t $. For the solutions to be positive we must have 
$\max(0,d_2-d_3)<t<\min(d_2,d_4) $. Hence, for this system of equation to always have  a positive solution we must have that  $ d _2 - d _3 < d _2 $, $ d _2 - d _3 < d _4 $ and $ d _2 , d _4 >0 $. 
The first inequality is always satisfied since $ d _3 >0 $.  The second inequality holds since $ d _5 -d _3 + d_2 - d_4=0 $, so that $ - d _3 + d _2  -d_4 =-d_5<0 $. Clearly,  $ d_4 >0$. The fact that
$ d _2 >0 $ follows from the fact that $ a _2>0 $, which completes the proof. 
\end{proof}

\section{Numerical Simulations}
\begin{figure}[!h]
\begin{center}
\subfigure[]{\rotatebox{0}{\includegraphics[width=0.45 \textwidth,
height=45mm]{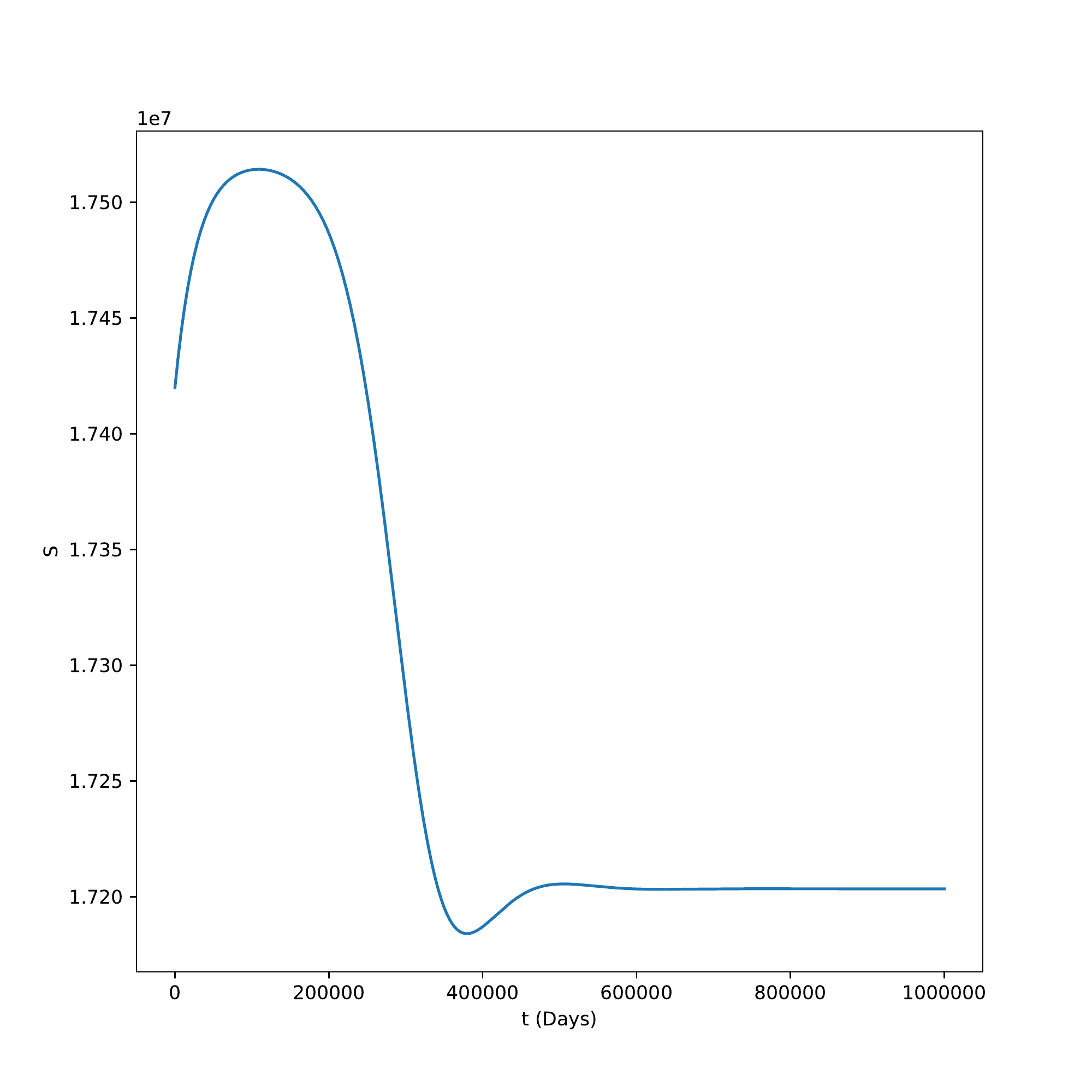}}}
\subfigure[]{\rotatebox{0}{\includegraphics[width=0.45 \textwidth,
height=45mm]{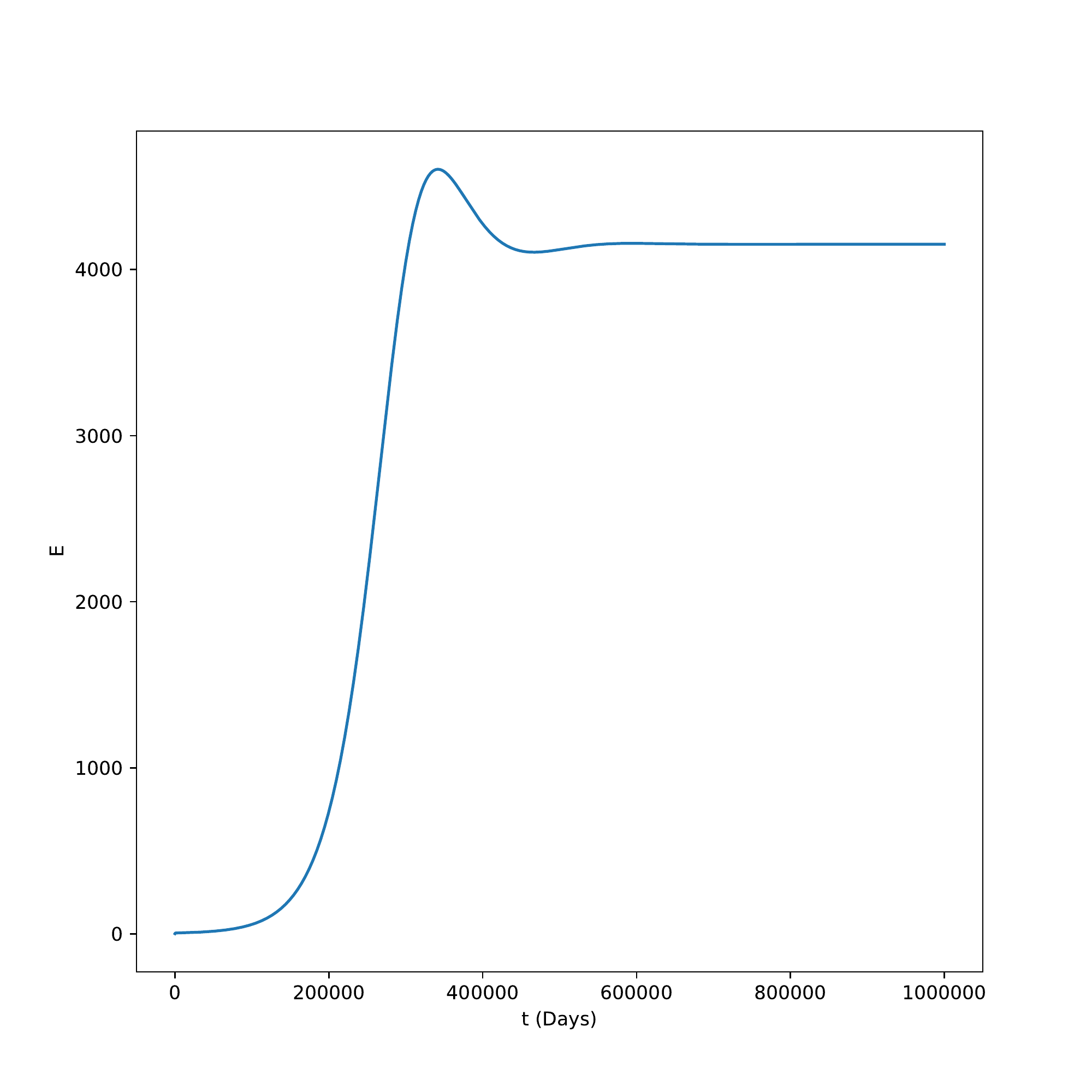}}}
\subfigure[]{\rotatebox{0}{\includegraphics[width=0.45 \textwidth,
height=45mm]{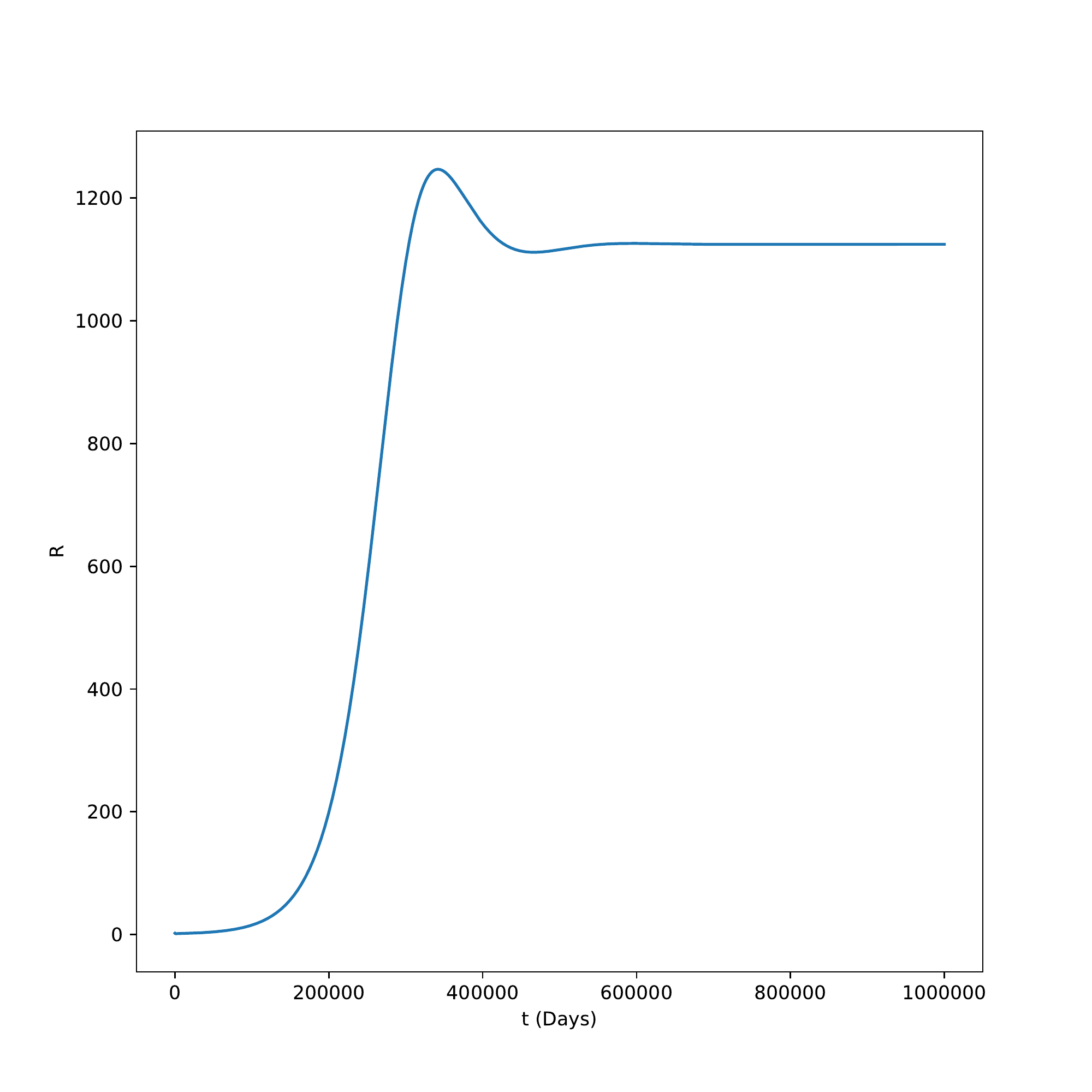}}}
\subfigure[]{\rotatebox{0}{\includegraphics[width=0.45 \textwidth,
height=45mm]{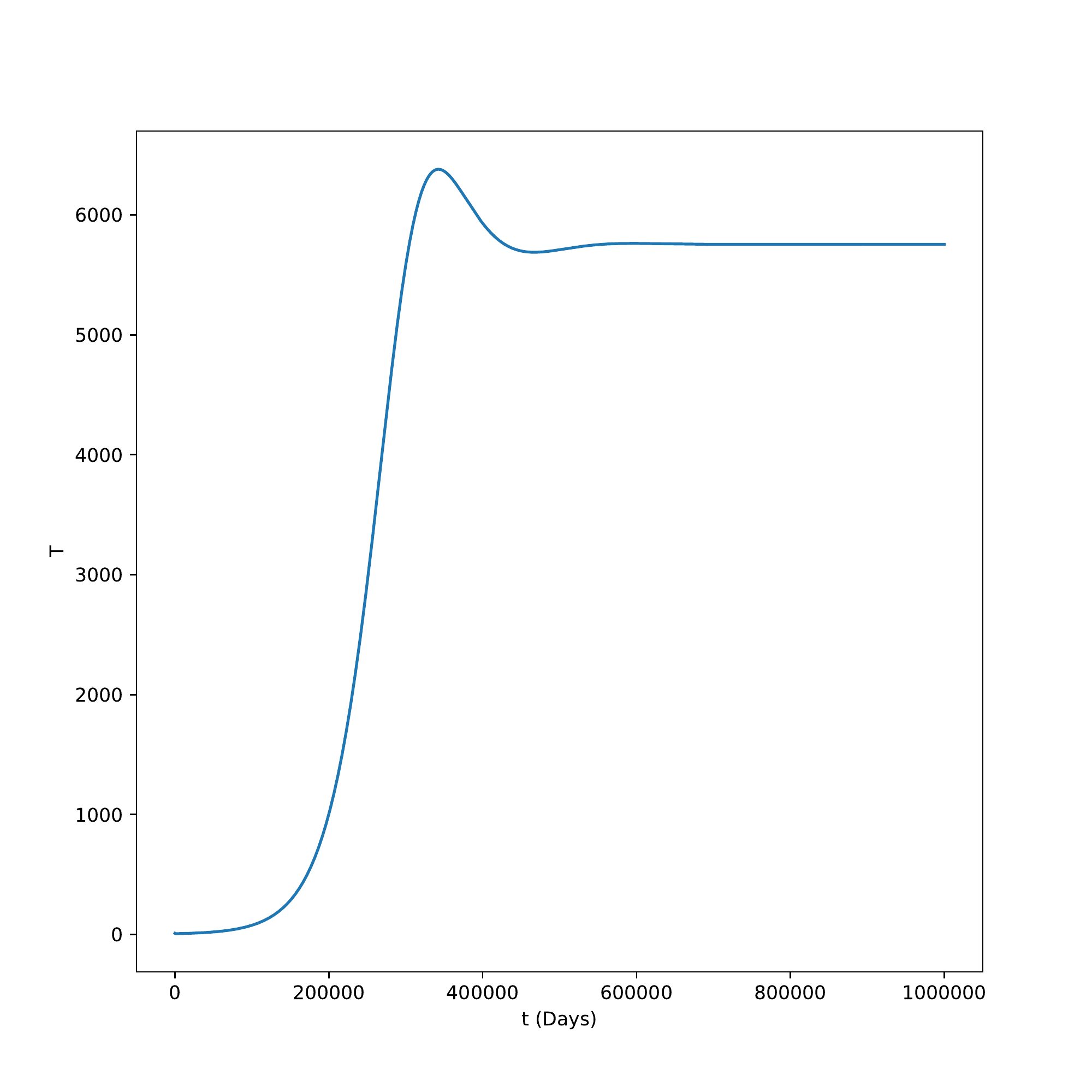}}}
\subfigure[]{\rotatebox{0}{\includegraphics[width=0.45 \textwidth,
height=45mm]{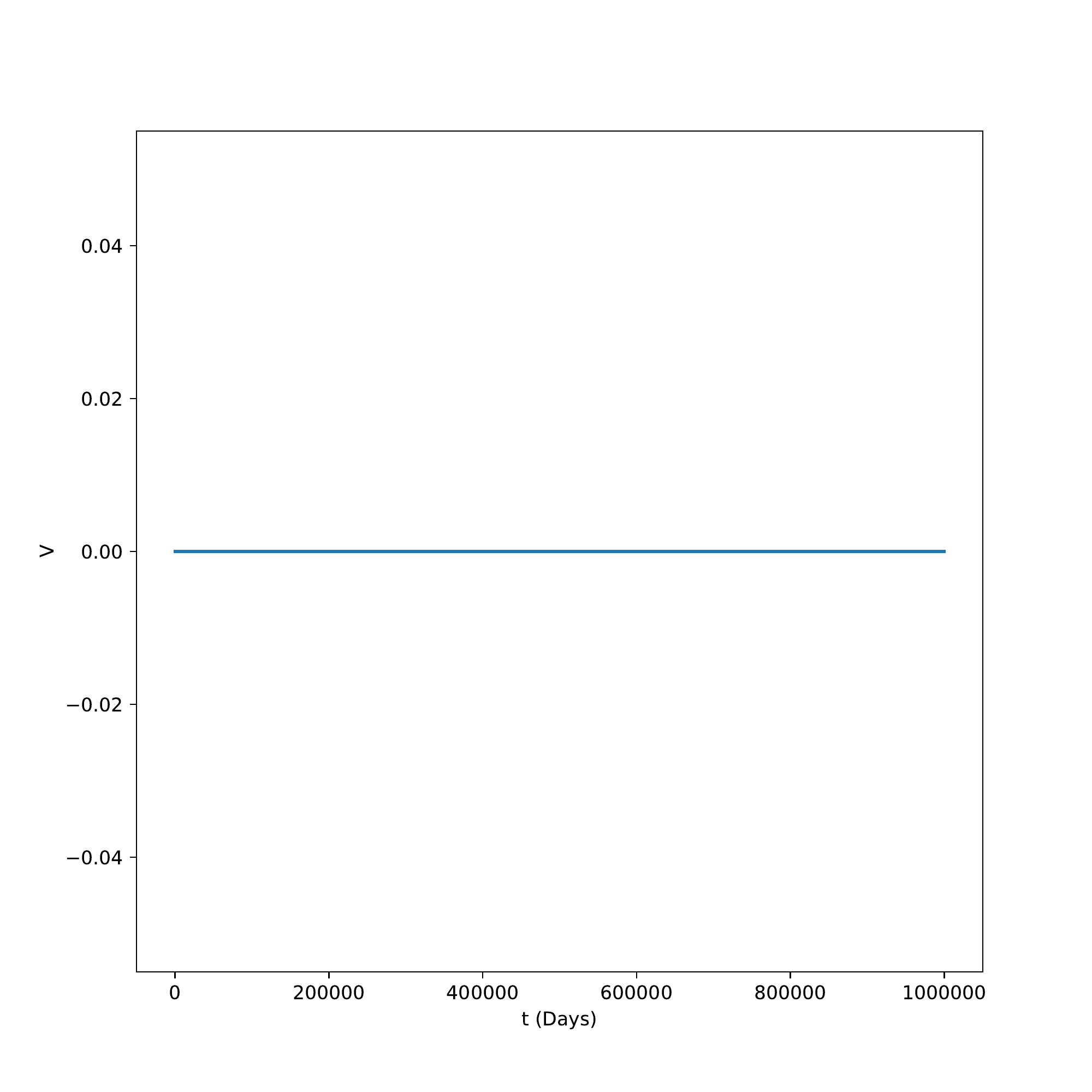}}}
\subfigure[]{\rotatebox{0}{\includegraphics[width=0.45 \textwidth,
height=45mm]{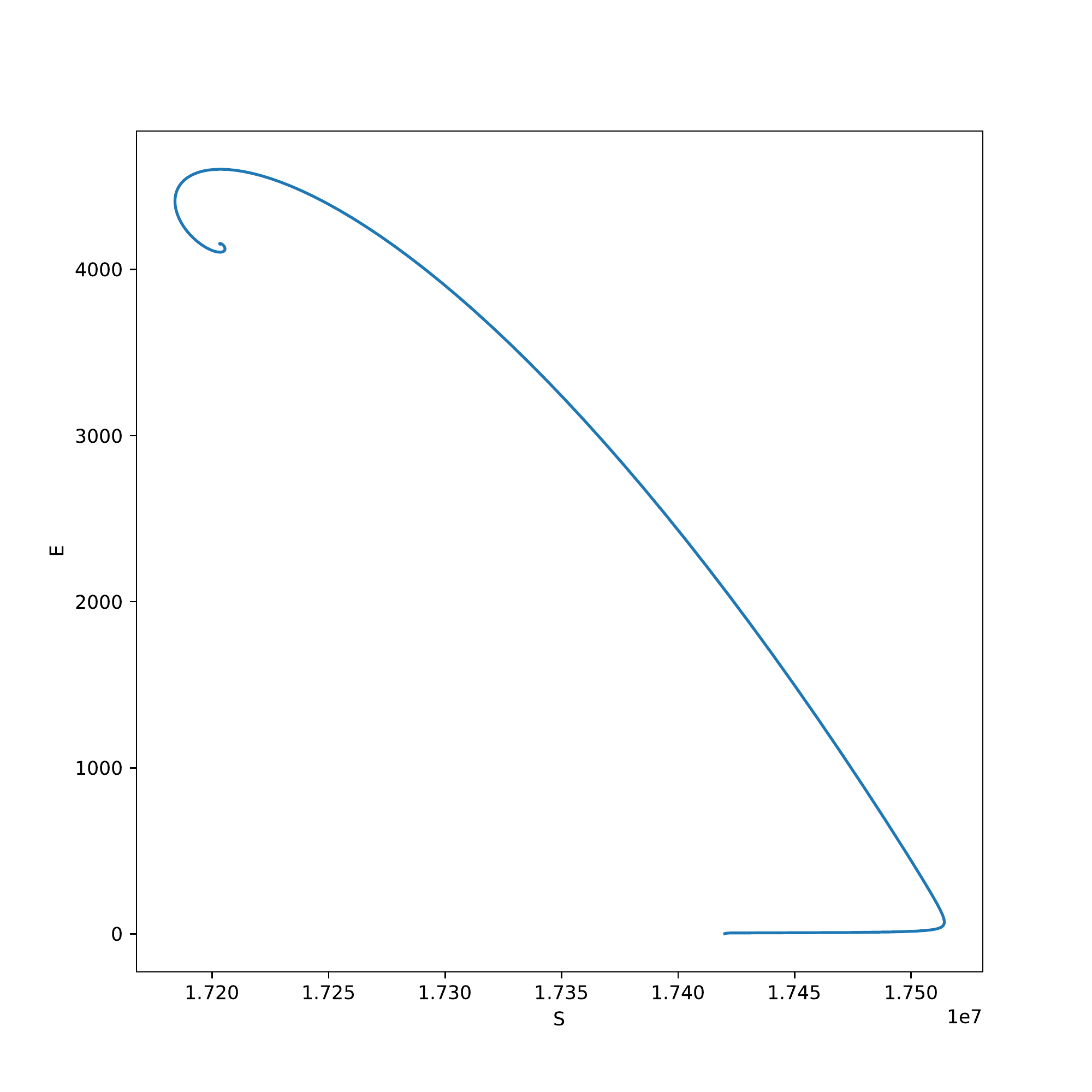}}}
 \vspace{-2mm}
 \caption{
\footnotesize Time history and phase portraits of system (\ref{eqn:modela}) for $\beta=0.0000005$,\,  
$q_E=0.86$,\, 
$q_R=0.14$,\,
$d_E=0.00083$,\,
$d_R=0.000083$,\,
$p_E= 0.00175$,\,
$p_R=0.0019$,\,
$p_S=1$,\,
$p_V=0$,\,
$\sigma=0.8$,\,
$c_E=0.0006$,\, 
$c_R=0.0008$,\, 
$k=0.66$,\,
$\delta=0.0016$,\, 
$\mu =0.000034247$,\,
$\Lambda=  600$ 
}\label{fig:Fig0}
\end{center}
 \end{figure}

We now use numerical simulations of system \eqref{eqn:modela} to illustrate and support the results of 
our mathematical analysis. 

We assume that the life expectancy is 80 years,  which implies that  death rate is $\mu = 0.000034247\,
(\textrm{days})^{-1}$ \cite{Hyman-James-M-and-LaForce-Tara2003}.
We take $\Lambda=600 \, (\textrm{days})^{-1}$, which corresponds to a population size of about 17.5 million.
The remaining  parameters are taken to  be $ \beta =0.0000005 \,(\textrm{days})^{-1}$,  
$d_E =0.00083 \,(\textrm{days})^{-1}$, $d_R=0.00083 \,(\textrm{days})^{-1}$, $p_E=0.00175$,
$p_R=0.0019 \,(\textrm{days})^{-1}$, $c_E=0.0006 \,(\textrm{days})^{-1}$, $c_R=0.0008 \,(\textrm{days})^{-1}$,  
and  $\delta=0.0016 \,(\textrm{days})^{-1}$, $q_E  = 0.86 $, $q_R=0.14 $ $ \sigma = 0.2 $ and $k=0.66$.
We also assume that, initially, there is no prevention program by choosing $ p _S = 1 $ , $ p _V = 0 $. 

In this case we find that $ \mathcal{R} _0 =1.0183912670368627$, and  hence, 
$ x ^\ast $ is globally asymptotically stable in $ \mathbb{R}  ^5 _{ >0 }$ by Theorem \ref{thm:x_ast}. 

Figures \ref{fig:Fig0} (a)-(d)  show that the number of individuals in the compartments $ S,E,R $, and $ T $,
  approach a constant value. 
 Figure \ref{fig:Fig0} (e) shows that the number of individuals in the vaccination compartment is, in this case zero. 
Figure  (f), instead, is a  phase portrait for the system. These figures  confirm that  the solutions approach a globally asymptotically stable equilibrium point. This case illustrates the   unwanted scenario where terrorists and recruiters become endemic to the population.

If we raise  $ p _V  $   to   $p_V=0.1 \,(\textrm{days})^{-1}$  without changing the other parameters, then  $ \mathcal{R} _0 =1.0126247813740998$ and the endemic equilibrium  $ x ^\ast $ is still globally asymptotically stable.
$ S, E, R $, and $ T $  still approach a constant value as $ t \to \infty $, however, the number of individuals in the compartments $ E $ and $ R $ is much less than before, see Figure \ref{fig:Fig1} (a)-(f). %Figures \ref{fig:E0} (a)--(d) show that $ S, E, R, T \to 0  $, as the time  $ t $ grows large, confirming that $ x _0 $ is globally asymptotically stable. This is  the preferred situation, where the number of extremists and recruiters decreases to zero.

If  we increase $ p _V  $   to $ p_V = 0.9 $, keeping the other parameters constant, then  $ \mathcal{R} _0 =0.9664928960719988$, so that 
$ x ^0 $ is globally asymptotically stable. Figures \ref{fig:Fig2} (a)--(d) show that $ S, E, R, V \to 0  $,
as   $ t \to \infty $, supporting our result on the asymptotic stability of  $ x_0 $.
This means that if a large enough number of susceptible individuals is targeted by effective
prevention programs it is possible to change the stability of the equilibria, ensuring that 
extremists and recruiters eventually disappear from the population.

\begin{figure}[!h]
\begin{center}
\subfigure[]{\rotatebox{0}{\includegraphics[width=0.45 \textwidth,
height=45mm]{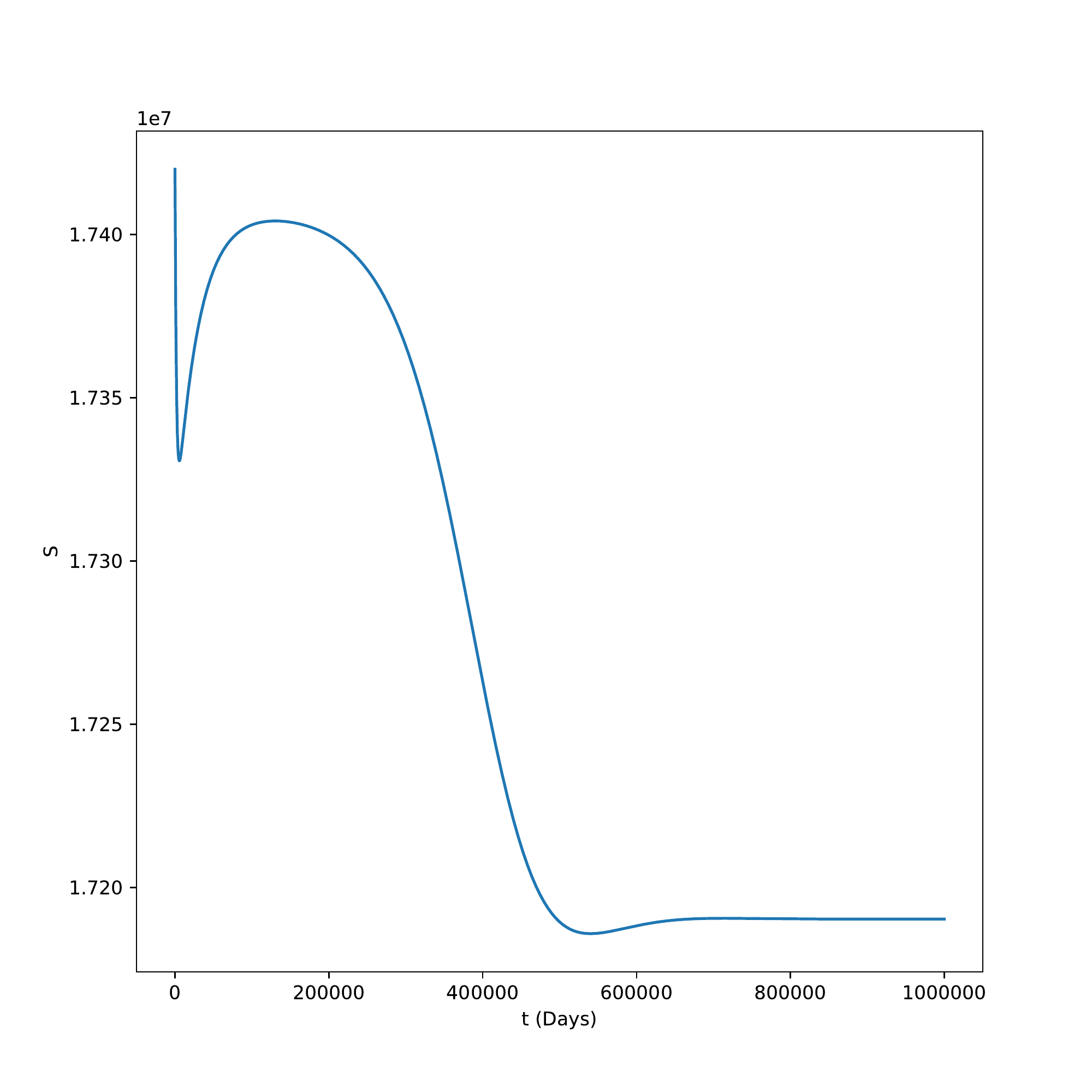}}}
\subfigure[]{\rotatebox{0}{\includegraphics[width=0.45 \textwidth,
height=45mm]{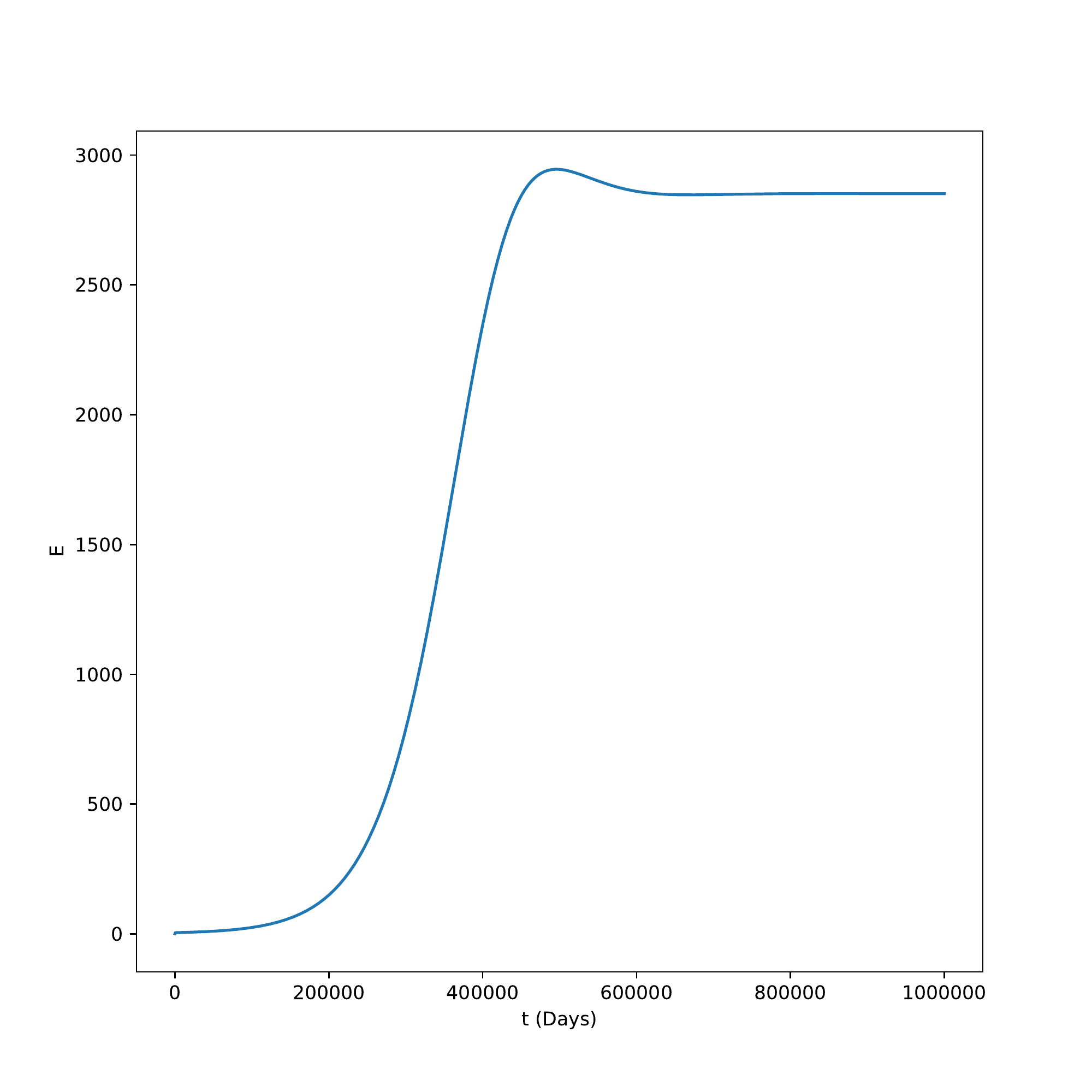}}}
\subfigure[]{\rotatebox{0}{\includegraphics[width=0.45 \textwidth,
height=45mm]{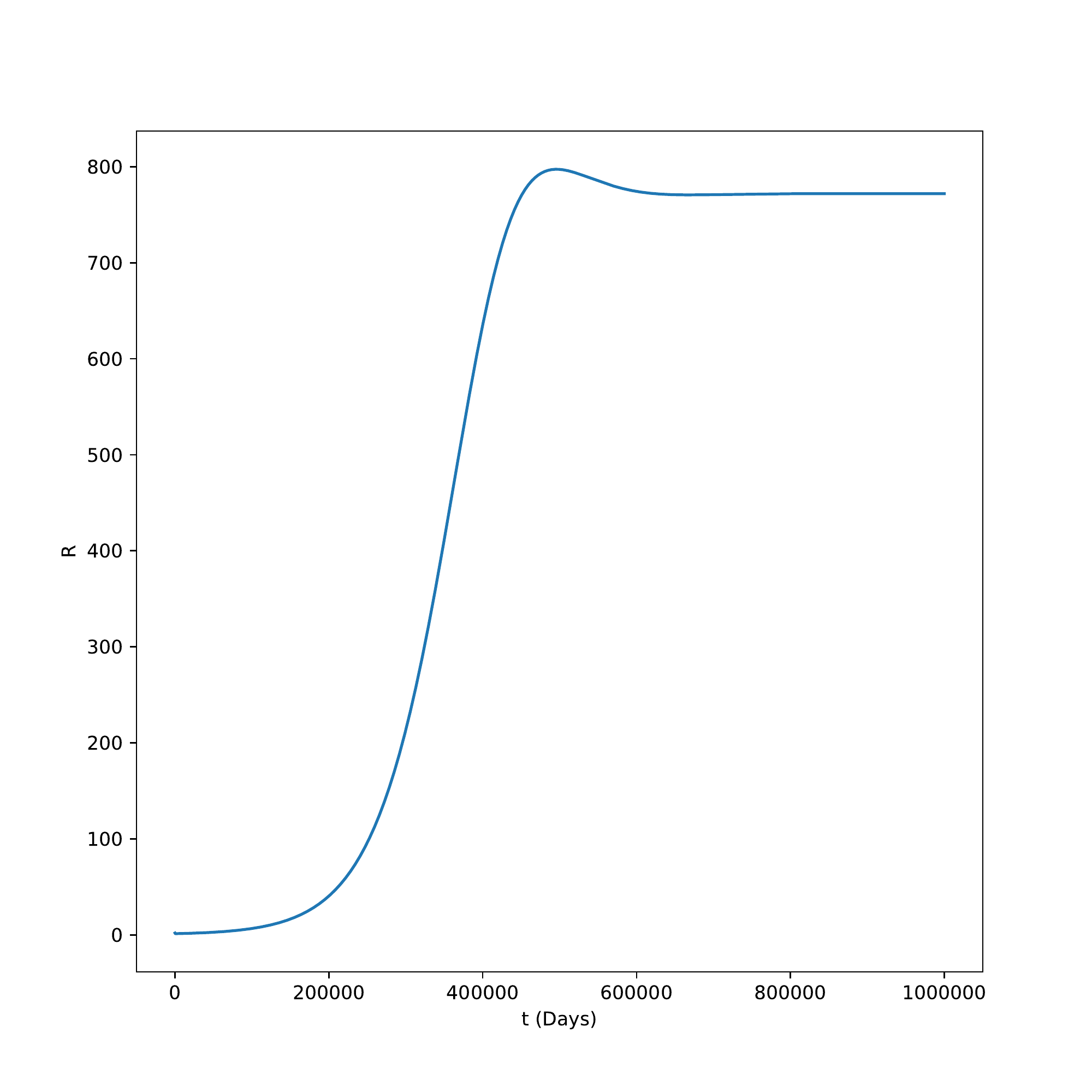}}}
\subfigure[]{\rotatebox{0}{\includegraphics[width=0.45 \textwidth,
height=45mm]{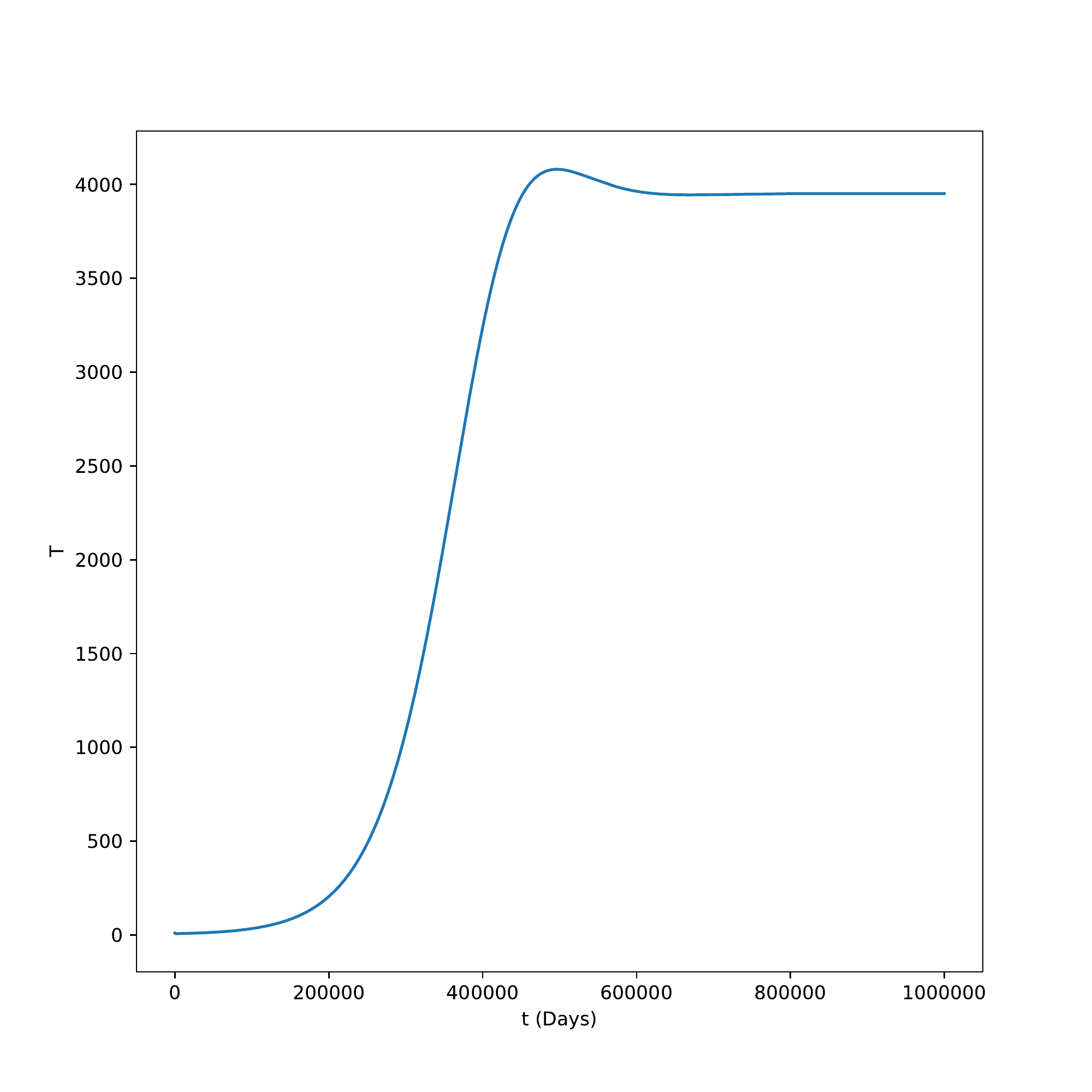}}}
\subfigure[]{\rotatebox{0}{\includegraphics[width=0.45 \textwidth,
height=45mm]{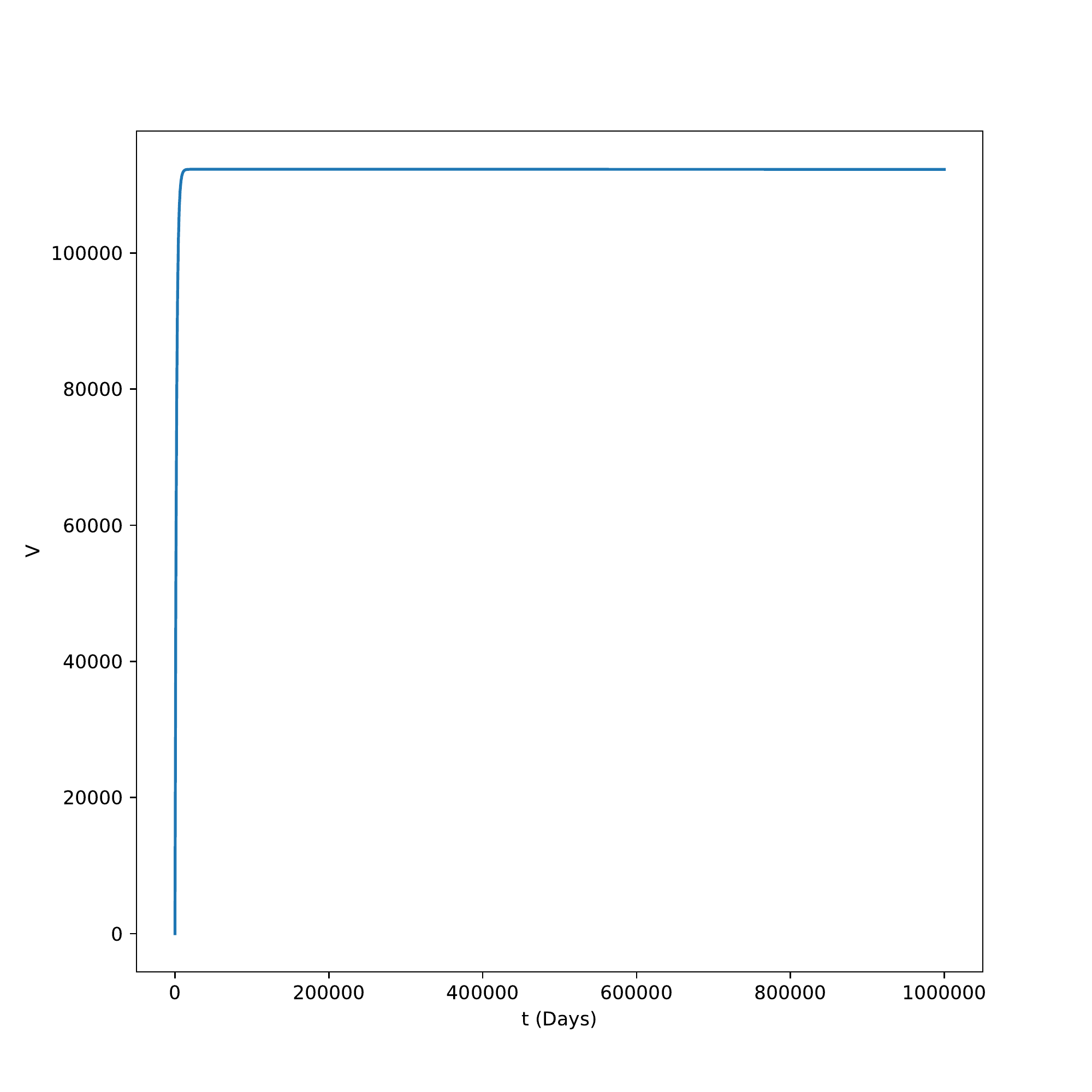}}}
\subfigure[]{\rotatebox{0}{\includegraphics[width=0.45 \textwidth,
height=45mm]{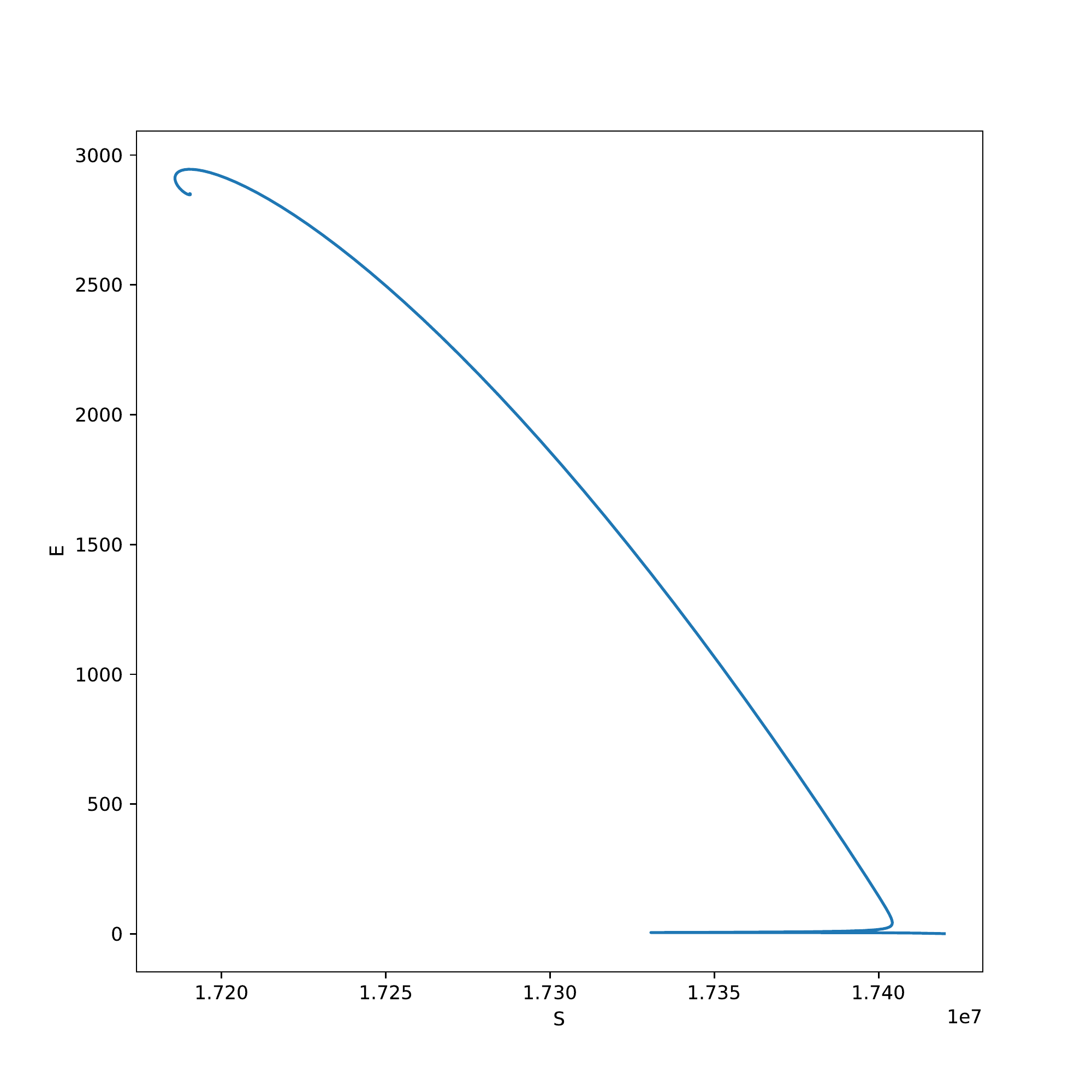}}}
 \vspace{-2mm}
 \caption{
\footnotesize Time history and phase portraits of system (\ref{eqn:modela}) for $\beta=0.0000005$,\,  
$q_E=0.86$,\, 
$q_R=0.14$,\,
$d_E=0.00083$,\,
$d_R=0.000083$,\,
$p_E= 0.00175$,\,
$p_R=0.0019$,\,
$p_S=0.9$,\,
$p_V=0.1$,\,
$\sigma=0.8$,\,
$c_E=0.0006$,\, 
$c_R=0.0008$,\, 
$k=0.66$,\,
$\delta=0.0016$,\, 
$\mu =0.000034247$,\,
$\Lambda=  600$ 
}\label{fig:Fig1}
\end{center}
 \end{figure}

 \begin{figure}[!h]
\begin{center}
\subfigure[]{\rotatebox{0}{\includegraphics[width=0.45 \textwidth,
height=45mm]{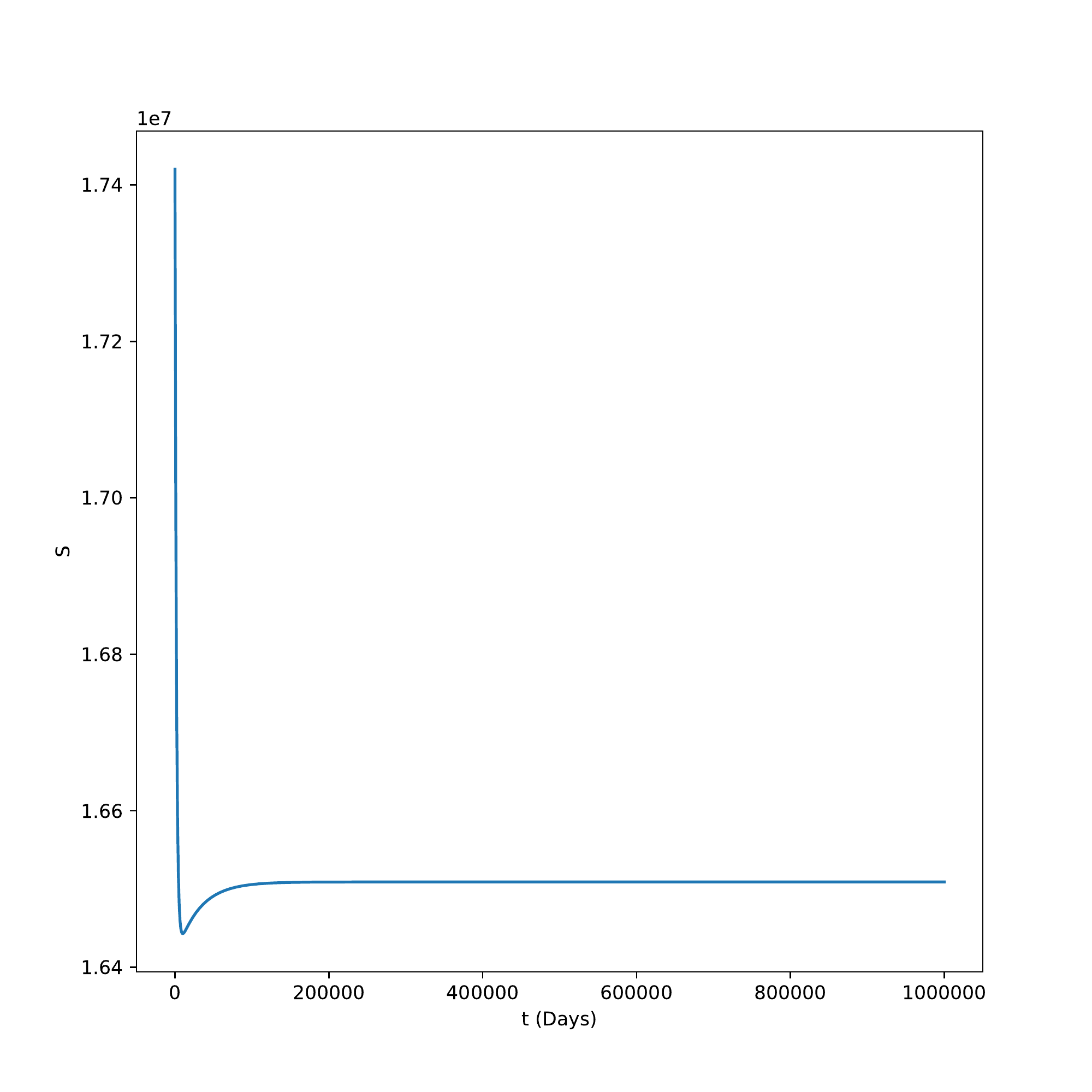}}}
\subfigure[]{\rotatebox{0}{\includegraphics[width=0.45 \textwidth,
height=45mm]{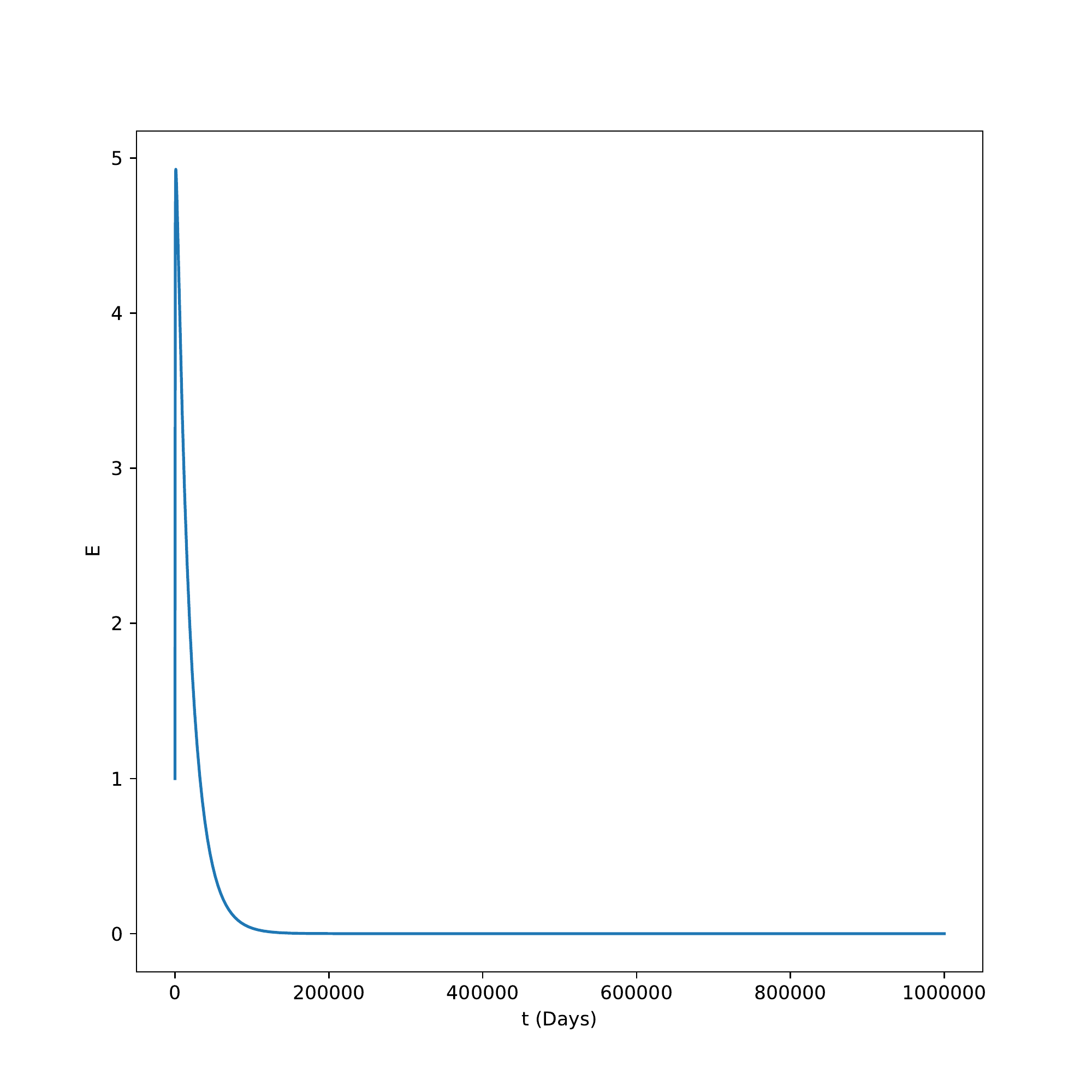}}}
\subfigure[]{\rotatebox{0}{\includegraphics[width=0.45 \textwidth,
height=45mm]{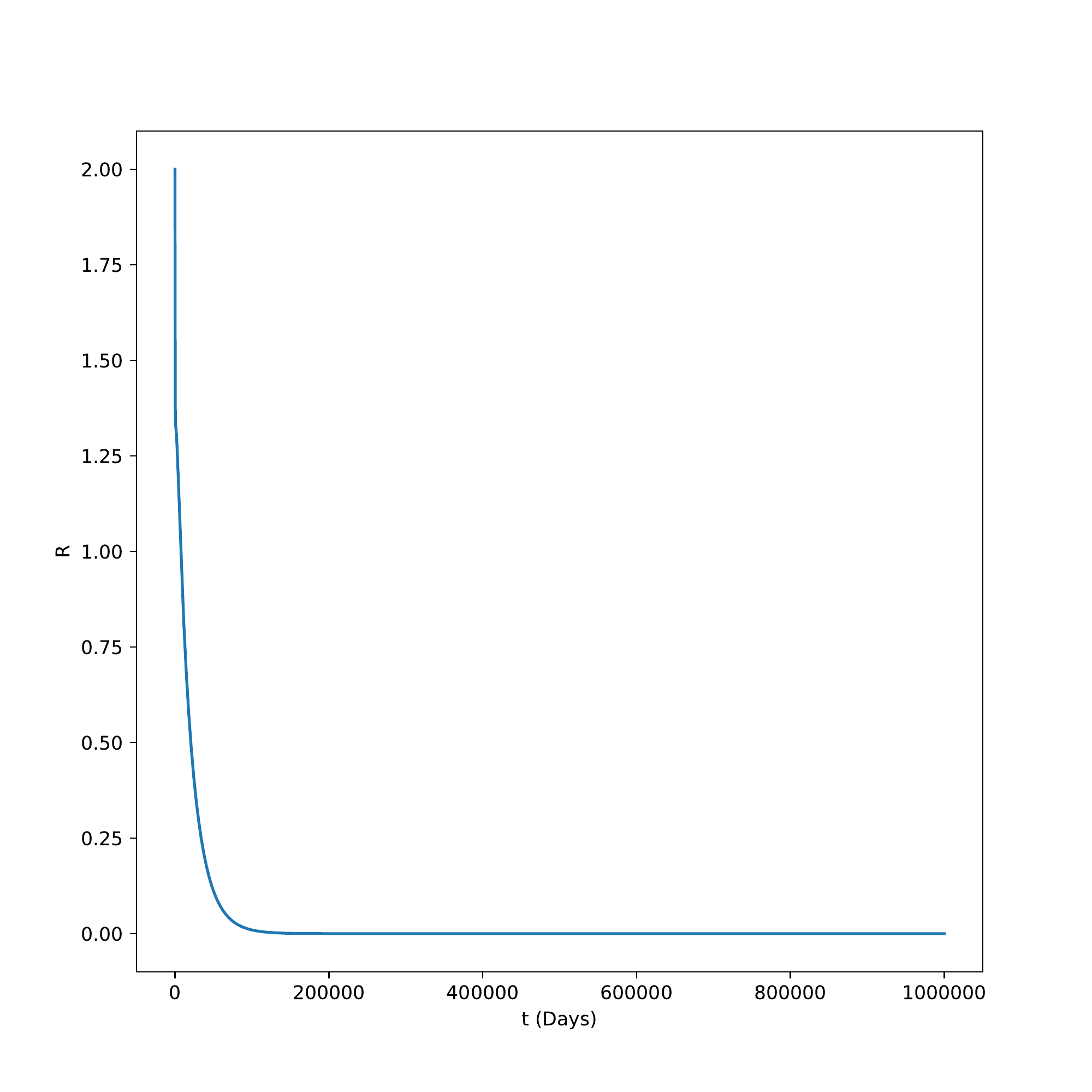}}}
\subfigure[]{\rotatebox{0}{\includegraphics[width=0.45 \textwidth,
height=45mm]{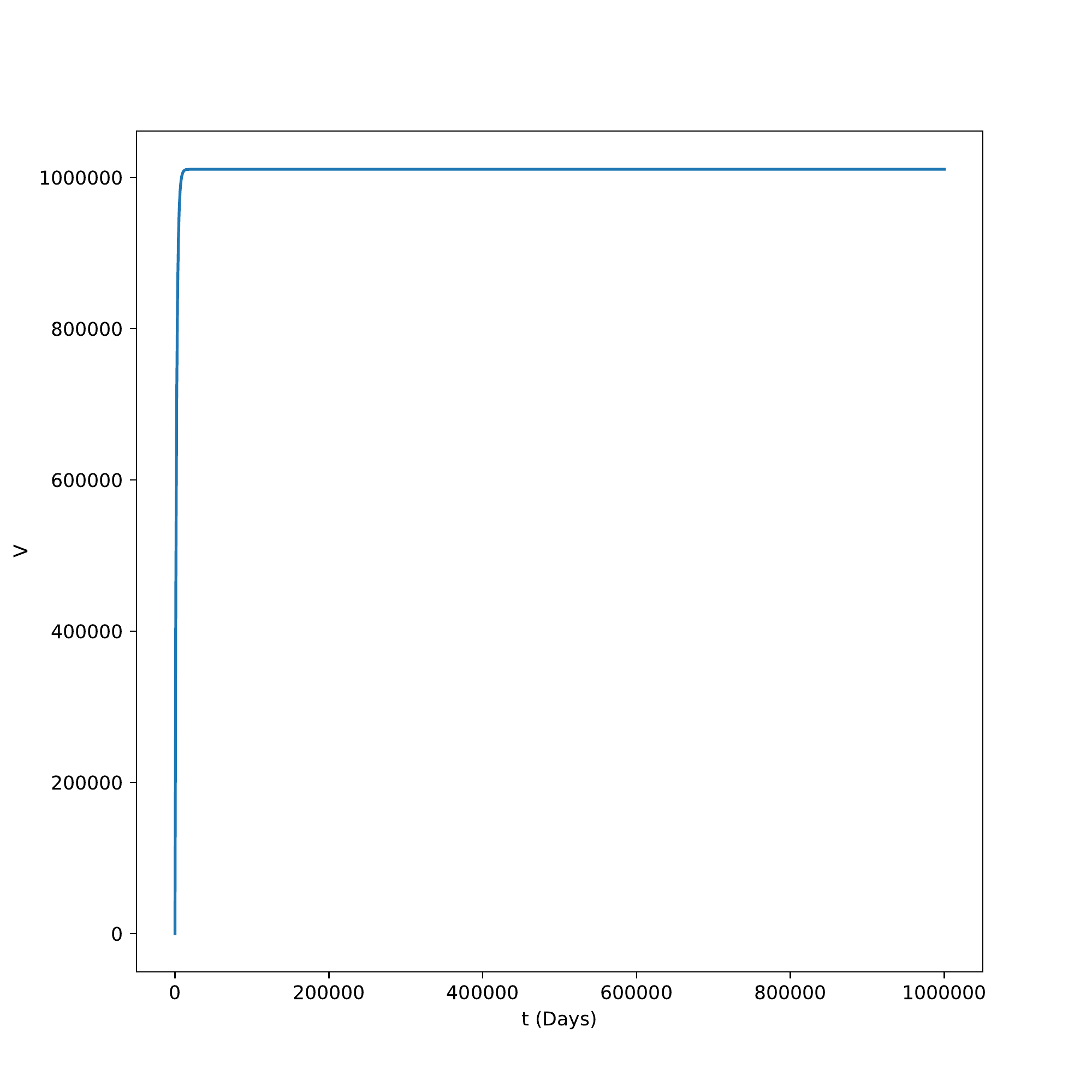}}}

 \vspace{-2mm}
 \caption{
\footnotesize Time history of system (\ref{eqn:modela}) for $\beta=0.0000005$,\,  
$q_E=0.86$,\, 
$q_R=0.14$,\,
$d_E=0.00083$,\,
$d_R=0.000083$,\,
$p_E= 0.00175$,\,
$p_R=0.0019$,\,
$p_S=0.1$,\,
$p_V=0.9$,\,
$\sigma=0.8$,\,
$c_E=0.0006$,\, 
$c_R=0.0008$,\, 
$k=0.66$,\,
$\delta=0.0016$,\, 
$\mu =0.000034247$,\,
$\Lambda=  600$ 
}\label{fig:Fig2}
\end{center}
 \end{figure}
 %\section{Discussion}
 \section{Conclusions}
This paper considers a model of radicalization in which the population is divided into five compartments: 
Susceptible, Vaccinated, Extremists, Recruiters and Treated. The model incorporates as part of the analysis
two key strategy of CVE, namely, prevention and de-radicalization. This work builds upon the papers of 
McCluskey and Santoprete \cite{mccluskey2018bare}, and  Santoprete and Xu \cite{santoprete2018global}.

The model considered in this paper has a threshold dynamics governed by the  basic reproduction number 
$ \mathcal{R} _0 $. If $ \mathcal{R} _0 \leq 1$, then there is only one equilibrium, 
free from  extremists and recruiters, which is globally asymptotically stable provided $ A p_V>4 $. In this case
the extremist ideology will be eradicated.
If $ \mathcal{R} _0 >1 $, the equilibrium mentioned above becomes unstable while an additional  equilibrium, 
which we call ``endemic", appears. This second equilibrium is globally asymptotically stable for
$ \mathcal{R} _0 > 1 $. In this case the ideology will become endemic, that is, recruiters and extremists will 
establish themselves in the population.

The analysis in  \cite{mccluskey2018bare} showed that    an increase in police and military action, that  is, 
increasing the parameters $ d _E $ and $ d _R $, decreased $ \mathcal{R} _0 $. 
The same conclusion applies in the model studied in this paper. 

In \cite{santoprete2018global} we established that an increase in the success rate of the de-radicalization 
programs, or an increase in the rates $ p _R $ and $ p _E $ at which individuals in the $ R $ and $ E $ 
compartments enter the Treated class caused $ \mathcal{R} _0 $ to decrease. Similar conclusions apply in 
the present model, proving that de-radicalization programs can be  an integral part  of a successful effort 
to combat violent extremism. If it is not possible to change $ k $, $ p _E $ or  $ p _R $ it  is often possible to  increase 
$ \frac{ 1 } { \delta } $,  the average prison sentence, which in turn decreases $\mathcal{R} _0 $. 
As observed in \cite{santoprete2018global}, this illustrates that incrementing  prison sentences length could 
be used to remedy the shortcoming of poorly designed de-radicalization programs.

Moreover, an increase in the fraction $ p _V $ of individuals entering the Vaccinated class 
produces a reduction in  $ \mathcal{R} _0 $. Similarly, a decrease in $ \sigma $, which accounts for the
decreased rate at which vaccinated individuals enter the Extremist class, corresponds to a decrease 
in $ \mathcal{R} _0 $. This demonstrates  that  prevention strategies can be effective  in reducing extremism.

As we have seen, since the basic reproduction number is expressed in terms of the model parameters, it is
easy  to evaluate various strategies to be used in combating violent extremism. 
However, there are two  shortcoming in our analysis. The first one  is that it is difficult to obtain realistic estimates of 
the parameters. This is, at least in part, due to a lack of empirical data regarding CVE 
programs \cite{mastroe2016surveying}.  The second is that we made many 
simplifying assumptions in the model. Some of these are listed in \cite{santoprete2018global}. 
These simplifying assumptions open up a number of  research 
paths that we intend to address in future work.

% Fakesection
%\Urlmuskip=0mu plus 1mu\relax
%\raggedright
%%% --- The following two lines are what needs to be added --- %%%
%\setcounter{biburllcpenalty}{800}
%\setcounter{biburlucpenalty}{800}
\bibliography{rad_papers}{}
\bibliographystyle{siam}
%\bibliographystyle{plain}
%\justifying

\end{document}